\documentclass[reprint,onecolumn]{article}
\usepackage[latin9]{inputenc}   
\usepackage{mathtools}
\usepackage{amstext}
\usepackage{amsthm}
\usepackage{graphicx}
\usepackage{dcolumn}
\usepackage{amsfonts} 
\usepackage{xcolor}
\usepackage{subfig}
\usepackage{authblk}
\usepackage[a4paper, margin=2cm]{geometry}
\usepackage{appendix}
\usepackage{booktabs} 
\usepackage{tikz-cd}
\usepackage{algorithm}
\usepackage{algpseudocode}
\usepackage{mathtools}
\usepackage{bbold}
\usepackage{subcaption} 
\makeatletter  

\newtheorem{theorem}{Theorem}[section]
\newtheorem{definition}[theorem]{Definition}  
\newtheorem{proposition}[theorem]{Proposition}
\newtheorem{lemma}[theorem]{Lemma}
\newtheorem{corollary}[theorem]{Corollary}
\newtheorem{remark}[theorem]{Remark}

\newtheorem{example}[theorem]{Example}

\usepackage{comment}

\newcommand{\mc}[1]{\mathcal{#1}}
\newcommand{\R}{\mathbb{R}}
\newcommand{\C}{\mathbb{C}}
\newcommand{\Z}{\mathbb{Z}}

\newcommand{\1}{\mathbb{1}}
\newcommand{\x}{\mathbf{x}}
\newcommand{\y}{\mathbf{y}}

\renewcommand{\u}{\mathbf{u}}
\renewcommand{\c}{\mathbf{c}}
\newcommand{\f}{\mathbf{f}}
\newcommand{\0}{\mathbf{0}}

\newcommand{\diag}{\textrm{diag}}
\newcommand{\sgn}{\textrm{sgn}}

\DeclareMathOperator\spec{Spec}
\DeclareMathOperator\tr{tr}

\date{}

\begin{document}

\title{Gremban Expansion for Signed Networks: Algebraic and Combinatorial Foundations for Community--Faction Detection}
\author[1,2]{Fernando Diaz-Diaz}
\author[3]{Karel Devriendt}
\author[3]{Renaud Lambiotte}
\affil[1]{Grupo Interdisciplinar de Sistemas Complejos (GISC), Departamento de Matematicas, Universidad Carlos III de Madrid, Leganes, Spain.}
\affil[2]{Institute of Cross-Disciplinary Physics and Complex Systems, IFISC (UIB-CSIC), 07122 Palma de Mallorca, Spain}
\affil[3]{Mathematical Institute, University of Oxford, OX2 6GG Oxford, United Kingdom.}
\date{}

\maketitle
\begin{abstract}
This article deals with the characterization and detection of community and faction structures in signed networks. We approach the study of these mesoscale structures through the lens of the Gremban expansion. This graph operation lifts a signed graph to a larger unsigned graph, and allows the  extension of standard techniques from unsigned to signed graphs. We  develop the combinatorial and algebraic properties of the Gremban expansion, with a focus on its inherent involutive symmetry. The main technical result is a bijective correspondence between symmetry-respecting cut-sets in the Gremban expansion, and regular cut-sets and frustration sets in the signed graph (i.e., the combinatorial structures that underlie communities and factions respectively). This result forms the basis for our new approach to community--faction detection in signed networks, which makes use of spectral clustering techniques that naturally respect the required symmetries. We demonstrate how this approach distinguishes the two mesoscale structures, how to generalize the approach to multi-way clustering  and discuss connections to network dynamical systems. 
\end{abstract}

\section{Introduction}

Signed networks, where edges are labeled as positive or negative, provide a framework to model antagonistic and cooperative interactions in complex systems \cite{altafini2012consensus,diaz2024mathematical,leskovec2010signed,szell2010multirelational}. Applications range from social networks exhibiting trust and conflict, to biological systems with activatory and inhibitory pathways, or to online platforms where user ratings may indicate agreement or opposition. Despite their relevance, signed networks remain more challenging to study than their unsigned counterparts \cite{traag2019partitioning}. Most foundational tools in graph theory and network science, particularly those based on distances or in spectral techniques, tend to be formulated, implicitly or not, under the assumption that edge weights are non-negative, thus limiting their direct applicability to signed structures \cite{diaz2025signed}.

This difficulty has led to the development of specialized tools for signed graphs, such as  spectral methods based on the eigenvectors of the signed graph's Laplacian \cite{cucuringu2019sponge,kunegis2010spectral}. While powerful, these approaches often require adapting existing frameworks or creating entirely new algorithms tailored to the signed context. An alternative and elegant route, similar to what has been done for other types of enriched network models \cite{lambiotte2019networks} (e.g. the supra-Laplacian for multiplex networks \cite{gomez2013diffusion}), consists in transforming the signed graph into a structure where standard, unsigned tools can be applied. For signed networks, the most natural such transformation is the Gremban expansion. This is a lifting technique that converts, without information loss, a signed graph into an unsigned one of larger size.

Originally introduced in the context of solving linear systems of equations \cite{gremban1996combinatorial}, the Gremban expansion provides a  bridge between signed and unsigned graph representations \cite{fox2018investigation,lambiotte2024spreading}. By lifting each node to two copies, referred to as positive and negative copy, and carefully connecting these copies according to the original edge signs, the expansion produces an unsigned graph whose structural and spectral properties retain all information of the signed graph. This opens the door to applying the extensive machinery of standard network techniques, including unsigned spectral graph theory \cite{chung1997spectral}, random walks \cite{Masuda2017RW}, and clustering techniques \cite{fortunato2010community}, to problems originally defined in the signed domain.

In this article, we explore the theoretical underpinnings 
and practical implications of the Gremban expansion for analyzing signed networks. We show how this transformation interacts with important structural features such as structural balance and graph operations such as sign switching, and demonstrate how symmetries in the expanded graph can give insights into the  connectivity patterns of the original signed graph. As an example of the power of the Gremban expansion, we develop a principled framework for distinguishing community and faction structures in signed graphs. Our algorithm leverages the symmetry properties of the eigenvectors of the expanded Laplacian to identify the dominant mesoscale organization. To our knowledge, this is the first method that jointly captures and separates both communities and factions in signed networks. We demonstrate its effectiveness through numerical experiments, show how it generalizes to multi-way clustering, and outline connections to dynamical processes on networks.  

The remainder of this article is structured as follows. Section \ref{sec:gremban_combinatorial} introduces the Gremban expansion and formalizes its structural and symmetry properties. Section \ref{sec:gremban_algebraic} develops the matrix formulation of the expansion, showing how adjacency and Laplacian matrices lift to Gremban-symmetric forms with spectral decompositions that separate community and factional eigenmodes. Section \ref{sec:spectral_detection_faction_community} presents our main algorithmic contribution: a spectral clustering method that leverages the Gremban Laplacian's eigenvectors to detect whether a network exhibits community or faction structure. We also report numerical experiments that validate our approach. Finally, in Section \ref{sec:other_applications} we outline connections to dynamical processes on signed networks, and conclude in Section \ref{sec:conclusion}.
%
\subsection{Related work}

The Gremban expansion of a signed graph also appears as a special case of some known constructions in the mathematical literature. First, signed graphs can be seen as a special case of so-called voltage graphs \cite{gross1987topological} used in topological graph theory. In that context, the Gremban expansion is called the ``derived graph" of a voltage graph. Secondly, the Gremban expansion can be seen as a special kind of graph covering \cite[\S6.8]{godsil2001algebraic}; this is further explained in Section \ref{SS: inverting Gremban expansion}. While spectral properties of graph covers have been studied in relation to graph expansion \cite{bilu2006lifts}, the interplay between the mesoscale structures of signed graphs and sign-induced graph covers is not treated there.

Beyond these algebraic and topological perspectives, the Gremban expansion has also appeared in the context of spectral clustering for signed networks. Spectral clustering is a widely used technique for identifying low-conductance cuts in networks, and its theoretical foundations in the unsigned case are well-established \cite{ng2001spectral, vonLuxburg2007}. Extensions to signed networks typically assume that the graph is close to structurally balanced and seek partitions that minimize the number of positive inter-group edges and negative intra-group edges---what we refer to in this work as factional structure. Building on this assumption, Kunegis \cite{kunegis2010spectral} proposed a spectral clustering algorithm that uses the eigenvectors of the signed Laplacian. Subsequent refinements, including those by Chiang et al. \cite{chiang2012scalable} and Cucuringu et al. \cite{cucuringu2019sponge}, further developed algorithms within this balance-centric framework. However, these approaches neglect standard community structure, where groups are defined by edge density regardless of sign. As a result, they are inherently biased in settings where factional and community structures coexist. 

Recent work by Fox et al. \cite{fox2017numerical,fox2018investigation} introduced the Gremban expansion in the context of signed graphs and explored its potential for clustering. In particular, \cite{fox2018investigation} compared clustering methods based on the signed Laplacian, the physical Laplacian, and the Gremban expansion, and found that the latter performed best in empirical tests. 
Nevertheless, this work did not treat communities and factions as distinct mesoscale structures, included no comparison with unsigned clustering baselines, and lacked a theoretical explanation for the empirical success of the Gremban method. Our work builds on theirs by developing a mathematically rigorous framework based on symmetry, which explicitly disentangles community and factional structures.
%
\section{The Gremban expansion of signed graphs}\label{sec:gremban_combinatorial}

The main objects of interest in this article are signed graphs. A signed graph $G=(V,E,\sigma)$ consists of a finite set $V$ of $n$ nodes, a finite set $E$ of $m$ edges, and a sign function $\sigma:E\rightarrow\{\pm 1\}$ on the edges. An unsigned graph is given by the data $G=(V,E)$. Throughout the article, self-loops, multi-edges and directed edges are excluded from consideration; each edge is thus a two-element subset of nodes, denoted $(u,v)\in E$ or $uv\in E$.

As discussed in the introduction, signed graphs naturally appear in applications. However, many tools in graph theory and network science, such as spectral and combinatorial methods, are limited to unsigned graphs and analyzing signed graphs often requires developing new tools and ideas. In this paper, we study the Gremban expansion as one possible way to work around this difficulty: by converting a signed graph into a larger unsigned graph (with structurally encoded sign data), we can use standard `unsigned methods' for the study of signed graphs. In particular, our article is motivated by the current methodological gap that exists in the study of mesoscale structures in signed graphs: how do signed factions and unsigned communities co-exist in a signed graph, and how might we detect and disentangle them?
%
\subsection{Gremban expansion: definition and symmetries}
The Gremban expansion, originally introduced and studied matrix theoretically in \cite{gremban1996combinatorial} and later graph theoretically in \cite{fox2017numerical}, is defined as follows:

\begin{definition}[Gremban expansion] \label{def:gremban_expansion}
Let $G=(V,E,\sigma)$ be a signed graph. The \emph{Gremban expansion} of $G$ is the unsigned graph $\mathcal{G}=\mathcal{G}(G)$ with $2n$ nodes and $2m$ edges, defined by
\begin{align*}
V(\mathcal{G})\,&:=\big\{ v^\chi \,:\, v\in V(G)\text{~and~}\chi\in\{\pm \} \big\}, 
\\
E(\mathcal{G})\,&:= \big\{ \big(u^\chi,v^{\chi\cdot\sigma(u,v)}\big) \,:\, (u,v)\in E(G) \text{~and~} \chi\in\{\pm \}\big\}.
\end{align*}
We refer to the two copies \( v^+ \) and \( v^- \)  of $v\in V$ in $\mathcal{G}$ as the {positive} and {negative polarities} of \( v \).
\end{definition}
In other words, the Gremban expansion doubles every vertex $v\in V(G)$ into two copies $v^+,v^-\in V(\mathcal{G})$ and doubles a positive edge $(u,v)$ into two edges between equal polarities $(u^\chi,v^\chi)$, and doubles a negative edge into two edges between opposite polarities $(u^{\chi},v^{-\chi})$. 

\begin{example}
Figure \ref{fig:first_example} illustrates the definition of Gremban expansion for a small signed graph. A $4$-cycle graph $G$ with one negative edge $(u,v)$ is expanded into an $8$-cycle graph. The example shows how positive edges (black in $G$) are lifted to edges between equal polarities in the Gremban expansion $\mathcal{G}$, whereas negative edges (red in $G$) are lifted to edges between opposing polarities in $\mathcal{G}$, i.e., $(u^+,v^-)$ and $(u^-,v^+)$.
\begin{figure}[h!]
    \centering
    \includegraphics[width=0.7\linewidth]{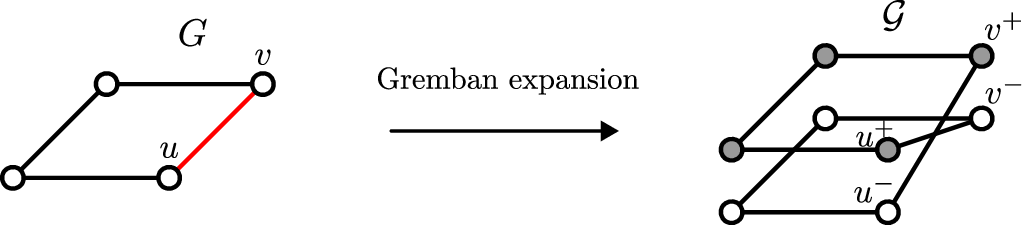}
    \caption{A signed $4$-cycle graph $G$ with one negative edge (in red) and its Gremban expansion $\mathcal{G}$. The positive polarities in $\mathcal{G}$ are shown as gray nodes, and the negative polarities as white nodes.}
    \label{fig:first_example}
\end{figure}
\end{example}
An important feature in the construction of the Gremban expansion is the structural symmetry between the two copies of a node; see Figure \ref{fig:first_example}. This symmetry plays a central role throughout this paper: it constrains the structure of expanded graphs, governs the behavior of random walks and diffusion processes, it is reflected in the spectral properties of matrices associated with the expansion, and it determines how spectral clustering on the expanded graph relates to faction and community structure in the original signed graph. We formalize this symmetry as follows; recall that an involution is a function $f$ that is its own inverse, i.e., with $f^2=\rm id$.

\begin{definition}[Gremban involution]\label{def:gremban_symmetry}
Let $\mathcal{G}$ be the Gremban expansion of a signed graph. The \emph{Gremban involution} is the involution $\eta:V(\mathcal{G})\rightarrow V(\mathcal{G})$ that permutes the two polarities of each node, $\eta:v^\chi\mapsto v^{-\chi}$.
\end{definition}
Since $\eta$ maps edges to edges (see further below), it can be extended to edges as $\eta((u,v))=(\eta(u),\eta(v))$ and to subsets of edges. We define corresponding notions of symmetries for subsets of nodes, subsets of edges and node partitions that are preserved under the Gremban involution:
\begin{definition}[Gremban symmetry]
Let $\mathcal{G}$ be the Gremban expansion of a graph, and $\eta$ its Gremban involution. Then we say that
\begin{itemize}
    \item A node subset $\mathcal{U}\subseteq V(\mathcal{G})$ is \emph{Gremban-symmetric} if $\eta(\mathcal{U})=\mathcal{U}$.
    
    \item An edge subset $\mathcal{F}\subseteq E(\mathcal{G})$ is \emph{Gremban-symmetric} if $\eta(\mathcal{F})=\mathcal{F}$.

    \item A subgraph $\mathcal{H}\subseteq\mathcal{G}$ is \emph{Gremban-symmetric} if $V(\mathcal{H})$ and $E(\mathcal{H})$ are Gremban-symmetric.

    \item A partition $V(\mathcal{G})=\mathcal{U}_1\cup\cdots\cup\mathcal{U}_k$ is \emph{Gremban-symmetric} if for all $i$, $\eta(\mathcal{U}_i)=\mathcal{U}_j$ for some $j$.
\end{itemize}
\end{definition}

\begin{example}
We return to the example in Figure \ref{fig:first_example} of the Gremban expansion $\mathcal{G}$ of the $4$-cycle with one negative edge. Figure \ref{fig:involution} shows the action of the Gremban involution $\eta$ on $\mathcal{G}$. Of the two labeled edges $e_1,e_2$, each one separately is not a Gremban-symmetric subset of edges (since, e.g., $\eta(e_1)=e_2\neq e_1$) but the pair of edges $\{e_1,e_2\}$ is Gremban-symmetric (since $\eta(\{e_1,e_2\})=\{e_1,e_2\}$).
\begin{figure}[h!]
    \centering
    \includegraphics[width=0.7\linewidth]{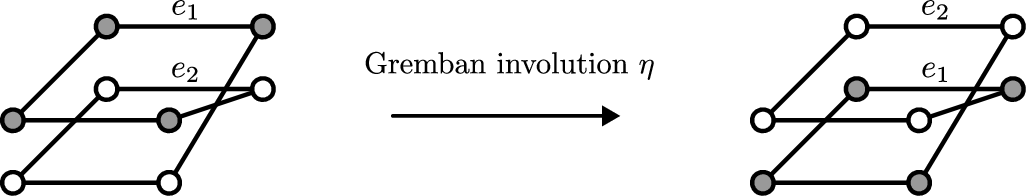}
    \caption{The Gremban involution $\eta$ switches polarities $v^{\chi}\mapsto v^{-\chi}$. A set of edges or nodes that is preserved under this action is called Gremban-symmetric. Here, $\{e_1,e_2\}\subset E(\mathcal{G})$ is Gremban-symmetric.}
    \label{fig:involution}
\end{figure}
\end{example}
The Gremban involution is related to more classical notions of graph symmetry, namely automorphisms. Recall that a \emph{graph automorphism} is an adjacency-preserving permutation of the nodes, in other words, a bijection \( \varphi: V \to V \) such that $(u, v) \in E \Leftrightarrow (\varphi(u), \varphi(v)) \in E$. Two nodes \( u, v \in V \) are said to be \emph{automorphically equivalent} if there exists an automorphism \( \varphi \) that maps one to the other, i.e., \( \varphi(u) = v \). An automorphism is called \emph{fixed-point-free} if it maps no node to itself. We find the following:

\begin{proposition}\label{prop:automorphically_equiv}
The Gremban involution \( \eta \) is a fixed-point-free involutive automorphism of the Gremban expanded graph. Every pair of nodes $v^\chi,v^{-\chi}$ is automorphically equivalent.
\end{proposition} 
\begin{proof}
By construction, $\eta$ is fixed-point-free and involutive. The Gremban expansion lifts each edge \( (u,v) \in E(G) \) to two edges in the expanded graph, and $\eta$ permutes the edges in this pair. Since every edge in $\mathcal{G}$ is the lift of an edge in $G$, $\eta$ acts as an involution on the edge set of $\mathcal{G}$. This means that $\eta$ preserves adjacency, which makes it an automorphism and thus certifies that $v^+, v^-$ are automorphically equivalent.
\end{proof}
As remarked in the introduction, the Gremban expansion also appears in the context of voltage graphs. More precisely, for voltage graphs with edge weights in $\Z/2\Z$, the Gremban expansion is called the derived graph and the Gremban involution is related to the so-called deck transformation group and corresponds to a natural symmetry inherited from the edge weight group $\Z/2\Z$; see \cite{gross1987topological} for more detail.

Proposition \ref{prop:automorphically_equiv} implies that the expanded graph must exhibit a high degree of symmetry between nodes with opposite polarities. As a result, any node property or associated structural quantity $f:V\rightarrow\R$ that is invariant under automorphisms will be equal for nodes with opposite polarities. This will play an important role when studying the spectral properties of Gremban expanded graphs.

\subsection{Inverting the Gremban expansion}\label{SS: inverting Gremban expansion}
A natural question is whether one can `invert' the Gremban expansion operation: can we recover the original signed graph $G$ from its Gremban expansion $\mathcal{G}$? For our purpose, this step is essential for interpreting results obtained in the expanded graph back in the signed setting. The \emph{projection map} introduced below will solve this recovery question; it acts by collapsing the expanded graph back onto the original node set. 
\begin{definition}[Projection map]
Let $\mathcal{G}$ be the Gremban expansion of a signed graph $G$. The \emph{projection map} $\pi:V(\mathcal{G})\rightarrow V(G)$ maps each polarity $v^\chi$ back to its original node $v$, as $\pi:v^\chi\mapsto v$.
\end{definition}
Let $\mathcal{N}(v^\chi)$ and $\mathcal{N}(v)$ denote the neighbours of $v^\chi$ in $\mathcal{G}$ and of $v$ in $G$, respectively. The projection map is a homomorphism since it maps edges to edges; using the language of graph covers \cite[\S6.8]{godsil2001algebraic} we can give the following more precise characterization:
\begin{proposition}  \label{prop:double_cover}
The projection map $\pi$ is a $2$-fold covering of $G$ by $\mathcal{G}$: for every $v\in V(G)$ and $v^\chi\in V(\mathcal{G})$
\begin{itemize}
    \item The fiber $\pi^{-1}(v)$ has exactly two elements: $\pi^{-1}(v)=\{v^{+},v^{-}\}$;
    \item $\pi$ is surjective and a local bijection: the restriction $\pi\vert_{\mathcal{N}(v^\chi)}:\mathcal{N}(v^\chi)\rightarrow\mathcal{N}(v)$ is a bijection.  
\end{itemize}
\end{proposition}
\begin{proof}
Surjectivity of $\pi$ is immediate by definition of the Gremban expansion. To prove that \( \pi|_{\mathcal{N}(v^\chi)} : \mathcal{N}(v^\chi) \to \mathcal{N}(v) \) is a bijection, we first observe that each neighbor \( w^\psi \in \mathcal{N}(v^\chi) \) is projected to \( \pi(w^\psi) = w  \in \mathcal{N}(v) \), so \( \pi \) maps \( \mathcal{N}(v^\chi) \) onto \( \mathcal{N}(v) \).
Suppose \( \pi|_{\mathcal{N}(v^\chi)} \) is not injective. Then there exist two distinct nodes \( w^+, w^- \in \mathcal{N}(v^\chi) \) such that \( \pi(w^+) = \pi(w^-) = w \). This means that both \( w^+ \) and \( w^- \) are adjacent to \( v^\chi \) in the expanded graph, which contradicts the construction of the Gremban expansion (Definition~\ref{def:gremban_expansion}). Therefore, \( \pi|_{\mathcal{N}(v^\chi)} \) must be injective, and hence bijective. 
\end{proof}
The practical use of the projection map \( \pi \) in this paper is to recover a signed graph from its Gremban expansion. We begin with the simplest case, the retrieval of the whole graph:
\begin{proposition}\label{prop:projection_graph}
Let $G$ be a signed graph, $\mathcal{G}$ its Gremban expansion, and $\pi$ the projection map. If we define \( \mu : E(\mathcal{G}) \to \{\pm 1\} \) as \( \mu((u^\chi, v^\psi)) = \chi \cdot \psi \), then $G = (\pi(V(\mathcal{G})), \pi(E(\mathcal{G})), \mu \circ \pi^{-1})$.
\end{proposition}
\begin{proof}
For the node and edge sets, we have \( \pi(V(\mathcal{G})) = V(G) \) and $\pi(E(\mathcal{G}))=E(G)$ by construction. It thus remains to show that $\mu$ is well-defined and satisfies \( \mu \circ \pi^{-1} = \sigma \). Let $e=(u,v)\in E(G)$ be any edge in $G$. Then \( \pi^{-1}(e) = \{(u^+, v^{\sigma(e)}), (u^-, v^{-\sigma(e)})\} \) and $\mu(\pi^{-1}(e))=\{+\cdot \sigma(e),-\cdot-\sigma(e)\}=\sigma(e)$ as required.
\end{proof}
This proposition shows that the projection map $\pi$ is the inverse of the Gremban expansion, and recovers the vertex and edge sets as well as the sign function of $G$ from $\mathcal{G}$. We now extend this result to subgraphs.

\begin{proposition}\label{prop:projection_subgraph}
Let $G$ be a signed graph, $\mc G$ its Gremban expansion, and $\pi$ the projection map. A subgraph $\mathcal{H}\subseteq \mathcal{G}$ is the Gremban expansion of a signed subgraph $H\subseteq G$ if and only if $\mathcal{H}$ is Gremban-symmetric.  In that case the corresponding signed subgraph is $H = (\pi(V(\mathcal{H})), \pi(E(\mathcal{H})), \sigma|_{\pi(E(\mathcal{H}))}).$
\end{proposition}
\begin{proof}
$(\!\Rightarrow\!)$  Suppose that $\mathcal{H}$ is the Gremban expansion of some signed subgraph $H\subseteq G$. 
Since each signed edge produces a symmetric pair in the expansion, the Gremban expansion $\mathcal{H}$ must be Gremban-symmetric.

$(\!\Leftarrow\!)$ Conversely, assume $\mathcal{H}\subseteq\mathcal{G}$ is a Gremban-symmetric subgraph, and define $H$ as in the Proposition. Since $\pi$ always maps edges of $\mathcal{G}$ to edges of $G$ and $E(\mathcal{H})\subseteq E(\mathcal{G})$, we have that $E(H)\subseteq E(G)$. In addition, the sign function on $H$ is constructed so that edges in $H$ have the same sign as the edges in $G$, so $H$ is a subgraph of $G$. We verify that the Gremban expansion of $H$ is $\mathcal{H}$. First, the Gremban expansion lifts $V(H)=\pi(V(\mathcal{H}))$ to $V(\mathcal{H})$. Second, the Gremban expansion lifts each edge $(v,w)\in E(H)=\pi(E(\mathcal{H}))$ to the edge pair $(v^+,w^{\sigma(vw)}),(v^-,w^{-\sigma(vw)})\in E(\mathcal{H})$ which, by definition of $\sigma$ and the Gremban symmetry of $\mathcal{H}$, coincides with $\pi^{-1}(u,v)$. This confirms that $\mathcal{H}$ is the Gremban expansion of $H$. 
\end{proof}
One further application of the projection map is the following characterization of Gremban graphs:
\begin{proposition}\label{prop:characterization of Gremban graphs}
An unsigned graph $\mathcal{G}$ is the Gremban expansion of a signed graph if and only if it has a fixed-point-free involutive automorphism $\eta$ such that $(u,v)\in E(\mc G)$ implies $v\neq \eta(u)$ and $(u,\eta(v))\not\in E(\mc G)$.
\end{proposition}
\begin{proof}
We will show how to construct a signed graph that lifts to $\mathcal{G}$, when an involution with the stated conditions exists. We first note that since $\eta$ is an automorphism, also $(\eta(u),v)\not\in E(\mathcal{G})$ and $\eta(u,v)\in E(\mathcal{G}$). As a result, for every edge $(u,v)\in E(\mc G)$ the induced subgraph on nodes $\{u,v,\eta(u),\eta(v)\}$ has precisely two edges $\{(u,v),\eta(u,v)\}$.
Now choose an antisymmetric polarization, i.e.,  $\chi:V(\mathcal{G})\rightarrow\{\pm\}$ such that $\chi\circ\eta=-\chi$. The projection map $\pi$ associated to this polarization then gives a signed graph, with sign induced by $\pi$, whose Gremban expansion is $\mathcal{G}$; indeed, each $4$-node induced subgraph described above projects to a signed edge, which in turn is lifted to the original subgraph. The converse follows from Proposition \ref{prop:automorphically_equiv} and Definition \ref{def:gremban_expansion}. 
\end{proof}
Proposition \ref{prop:characterization of Gremban graphs} shows that a Gremban graph is specified by a triple $(\mathcal{G},\eta,\chi)$ consisting of:
\begin{itemize}
    \item an unsigned graph $\mathcal{G}$;
    \item an involution $\eta$ of $\mathcal{G}$ that satisfies the conditions in Proposition \ref{prop:characterization of Gremban graphs};
    \item an antisymmetric polarization $\chi$ of $V(\mathcal{G})$.
\end{itemize}

In some cases, we are interested in recovering only the positive or negative polarity of nodes or edges from the Gremban expansion. To this end, we define the one-sided projection map:  
\begin{definition}[One-sided projection map]  \label{def:one_sided_projection}
    Let \( \mathcal{G} \) be the Gremban expansion of a signed graph \( G \). The \emph{one-sided projection map} \( \pi_\chi : V(\mathcal{G}) \to V(G) \cup \{ \emptyset \} \) is defined by $\pi_{\chi}:v^{\chi}\mapsto v$ and $\pi_{\chi}:v^{-\chi}\mapsto \emptyset$.
\end{definition}
The one-sided projection map extends to edges in the standard way for edges between the correct polarities, i.e., $\pi_{\chi}:(u^\chi,v^\chi)\mapsto(u,v)$, and it maps all other edges to $\emptyset$. This map furthermore interacts nicely with Gremban-symmetric sets and partitions.

\begin{proposition} \label{prop:one_sided_projection}
Let \( \mathcal{U} \subseteq V(\mathcal{G}) \) be a Gremban-symmetric node set. Then \( \pi_+(\mathcal{U}) = \pi_-(\mathcal{U}) \). Moreover, if $\mathcal{U}_1\cup\cdots\cup\mathcal{U}_k$ is a Gremban-symmetric partition of $V(\mathcal{G})$, then $\pi_\chi(\mathcal{U}_1)\cup\cdots\cup\pi_{\chi}(\mathcal{U}_k)$ is a partition of $V(G)$.
\end{proposition}
\begin{proof}
Since \( \mathcal{U} = \eta(\mathcal{U}) \), every node \( v^\chi \in \mathcal{U} \) implies \( v^{-\chi} \in \mathcal{U} \) as well. Therefore,
\[
\pi_+(\mathcal{U}) = \{ v \in V(G) : v^+ \in \mathcal{U} \} = \{ v \in V(G) : v^- \in \mathcal{U} \} = \pi_-(\mathcal{U}).
\]
Now let $\mathcal{U}_1\cup\cdots\cup\mathcal{U}_k$ be a Gremban-symmetric partition of \( V(\mathcal{G}) \). By construction, each node \( v \in V(G) \) has both \( v^+ \) and \( v^- \) appearing in exactly one \( \mathcal{U}_i \), so \( v \in \pi_+(\mathcal{U}_i) \) for a unique \( i \). This ensures that the sets \( \pi_+(\mathcal{U}_i) \) are disjoint and collectively cover \( V(G) \), hence forming a partition.
\end{proof}
%
\subsection{Sign switching}

As illustrated in the previous section, one of the key features of the Gremban expansion is its ability to reflect the sign information of the original graph in the combinatorial structure of the expanded graph. This observation translates nicely when considering operations between signed graphs. In this section, we study how certain sign switching operations \cite{zaslavsky1982signed} on signed graphs translate structurally to permutations of their Gremban expansions. 
\begin{definition}
Let $G=(V,E,\sigma)$. A switching function is a map $\theta:V\rightarrow\{\pm 1\}$ and this induces:
\begin{itemize}
    \item A \emph{switching partition} $V=U^+_\theta\cup U^-_\theta$, where $U^-_\theta :=\{v\in V : \theta(v)=-1\}$ and $U^+_\theta= V\setminus U_\theta^-$.
    \item A \emph{switched graph} $G_{\theta} = (V,E,\sigma_\theta)$ where $\sigma_\theta(u,v)=\theta(u)\theta(v)\sigma(u,v)$.
\end{itemize}
\end{definition}
We start by observing that switching functions can always be decomposed into \emph{elementary switching functions}: these are switching functions $\theta_v$ labeled by $v\in V$, and defined by $\theta_v(v)=-1$ and $+1$ otherwise.
\begin{lemma} \label{lem:vertex_vs_global_switching}
Let $G$ be a signed graph and $\theta$ a switching function. Then switching $G$ with respect to $\theta$ is equivalent to switching $G$ consecutively with respect to the elementary switching functions $\theta_{v_1},\dots,\theta_{v_k}$, where $U_\theta^{-}=\{v_1,\dots,v_k\}$ are those nodes for which $\theta(v)=-1$, arranged in arbitrary order. 
\end{lemma}

\begin{proof}
We prove that switching with respect to \( \theta \) is equivalent to performing successive switchings at all nodes in \( U^-_\theta=\{v:\theta(v)=-1\}\).
Consider an arbitrary edge \( uv \in E \), and track how its sign changes as each switching is applied. If neither \( u \) nor \( v \) belongs to \( U^-_\theta \), then the edge is not affected. If exactly one of the endpoints lies in \( U^-_\theta \), then the sign of the edge is flipped once. If both endpoints are in \( U^-_\theta \), then the edge is flipped twice, and the sign remains unchanged. In all cases, the resulting sign of \( uv \) is \( \theta(u) \theta(v) \sigma(uv) \), which is the transformation defined by switching with respect to \( \theta \). 
\end{proof}
We now show how the switching of a signed graph is reflected in its Gremban expansion:
\begin{proposition}  \label{prop:switching_permutation}
Let \( G = (V,E,\sigma) \) be a signed graph, $\mathcal{G} = \mathcal{G}(G)$ its Gremban expansion, and $\theta$ a switching function. Let $\tau$ be the permutation of $V(\mathcal{G})$ defined by $\tau(v^{\chi})=v^{-\chi}$ for all $v$ for which $\theta(v)=-1$, and leaving all other nodes of $\mathcal{G}$ unchanged. Then switching $G$ according to $\theta$ corresponds to permuting the vertex polarities of $\mathcal{G}$ according to $\tau$. In other words, the following diagram commutes:
$$
    \begin{tikzcd}
    G \arrow{r}{\text{Gremban}} \arrow[swap]{d}{\theta} &[2em] \mathcal{G}(G) \arrow{d}{\tau} 
    \\
    G_\theta \arrow{r}{\text{Gremban}} & \mathcal{G}(G_{\theta})
    \end{tikzcd}
$$
In particular, switching $G$ with respect to $\theta_v$ corresponds to interchanging the two polarities of $v$ in $\mathcal{G}$.
\end{proposition}

\begin{proof}
By the composition property of Lemma \ref{lem:vertex_vs_global_switching}, it suffices to prove the proposition for elementary switching functions. Let $(u,v)$ be a positive edge. Then \( \mc G(G) \) contains the edges \( (v^+, w^+) \) and \( (v^-, w^-) \). After switching \( v \), the edge becomes negative, and its expansion contains \( (v^+, w^-) \) and \( (v^-, w^+) \), which are exactly the result of swapping \( v^+ \leftrightarrow v^- \). The same logic applies if \( (v, w)\) is a negative edge. Hence, \( \mc G(G_{\theta_v}) = \tau_v \mc G(G) \tau_v^{-1} \), and the diagram commutes.
\end{proof}
In other words, inverting the sign of all edges incident on a node corresponds to permuting the polarities of that node in the Gremban expansion. More generally, inverting the sign of all edges between two node groups in a partition corresponds to permuting the polarities of one of the groups in the Gremban expansion. For a given signed graph $G$, the switching operation can be thought of as some form of `gauge equivalence'; this idea will come back in the next section when we optimize a function of $G$ over all possible switchings.
    
\begin{example}
To illustrate Proposition~\ref{prop:switching_permutation} we revisit the signed $4$-cycle graph. Figure \ref{fig:commutative_diagram} shows the effect of switching signs at one of the nodes on the signed graph and on the Gremban graph.
\begin{figure}[h!]
    \centering
    \includegraphics[width=0.65\linewidth]{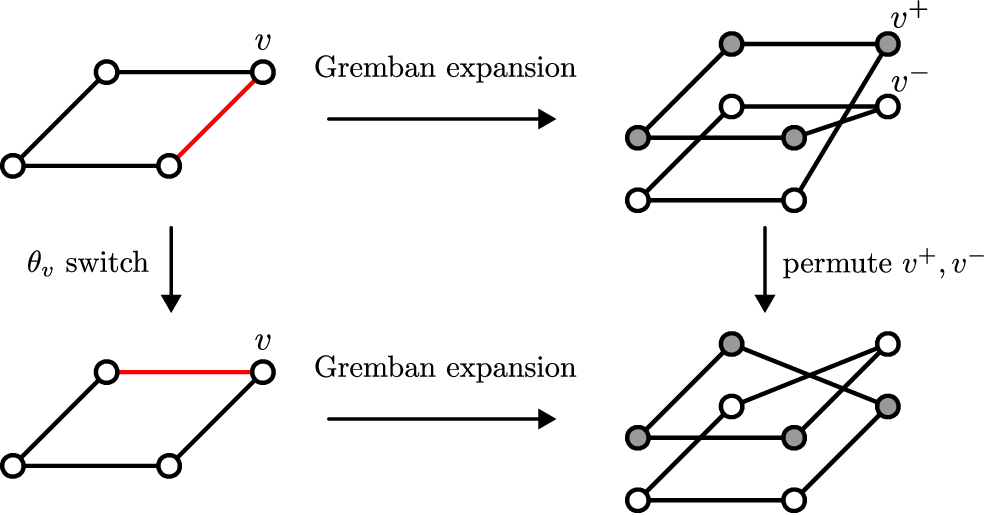}
    \caption{Switching the signs of a signed graph by a switching function translates to permuting nodes in the Gremban expansion. Here, the elementary switching by $\theta_v$ in $G$ results in the permutation of $v^{+},v^{-}$ in $\mathcal{G}$. In the bottom-right graph, the color of $v^+,v^-$ reflects the polarity of these nodes in $\mathcal{G}(G)$, before permuting.}
    \label{fig:commutative_diagram}
\end{figure}
\end{example}

\subsection{Balance and connectivity}
One important property of signed graphs is structural balance \cite{harary1953notion,cartwright1956structural}: a signed graph $G$ is called \emph{balanced} if there is a bipartition $V(G)=U_1\cup U_2$ of its nodes such that all edges within $U_1,U_2$ are positive and all edges between $U_1$ and $U_2$ are negative. The subsets $U_1,U_2$ are called the \emph{balanced factions} of $G$. We first consider (exact) balance and then move to approximate balance and how to quantify it. We start with an example that illustrates how the Gremban expansion of a signed graph can encode its balance.
\begin{example}
Figure \ref{fig:first_example} shows that the Gremban expansion of an unbalanced $4$-cycle is an (unsigned) \textit{connected} $8$-cycle. Figure \ref{fig:balanced_cycle} shows that for a balanced $4$-cycle, the Gremban expansion is two \textit{disconnected} copies of a $4$-cycle. This observation is made precise in Theorem \ref{thm:balanced_disconnected}.
\begin{figure}[h!]
    \centering
    \includegraphics[width=0.6\linewidth]{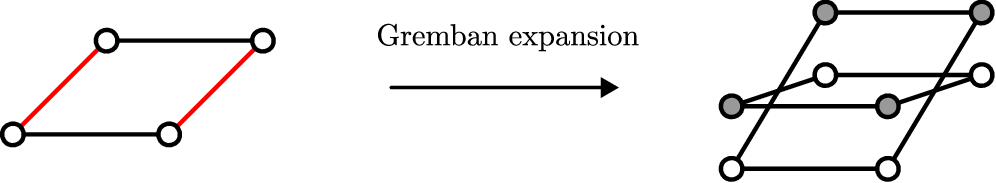}
    \caption{The Gremban expansion of a balanced signed $4$-cycle is two disconnected copies of a $4$-cycle.}
    \label{fig:balanced_cycle}
\end{figure}
\end{example}
Balance in signed graphs knows many characterizations and dates back at least to the work of Harary in the 1950s~\cite{harary1953notion} (see also \cite{babul2025strong,diaz2025mathematical}). We now give a characterization in terms of the Gremban expansion:
\begin{theorem}  \label{thm:balanced_disconnected}
A connected, signed graph is balanced if and only if its Gremban expansion is disconnected.
\end{theorem}

\begin{proof}
Let \( G \) be a connected and balanced signed graph. Then there exists a switching function \( \theta \) such that \( G_\theta \) is all-positive, namely the switching function that assigns $-1$ to one of the balanced factions and $+1$ to the other \cite{diaz2025mathematical}. Its Gremban expansion is therefore the disjoint union \( \mathcal{G}(G_\theta) = G_\theta \sqcup G_\theta \), which is disconnected. By Proposition~\ref{prop:switching_permutation}, \( \mathcal{G}(G) \) and \( \mathcal{G}(G_\theta) \) are related by a permutation of nodes, so \( \mathcal{G}(G) \) is also disconnected.

Conversely, assume \( G \) is unbalanced. Choose a spanning tree \( T \subseteq G \), and define a switching \( \theta \) that makes all edges in \( T \) positive. Then \( G_\theta \) contains a positive spanning tree, so its expansion \( \mathcal{G}(G_\theta) \) contains two connected subgraphs isomorphic to \( T \). Since \( G \) is unbalanced, there exists at least one negative edge in \( G_\theta \setminus T \), and its expansion connects the two copies. Hence, \( \mathcal{G}(G_\theta) \), and thus \( \mathcal{G}(G) \), is connected.
\end{proof}
We can extend the idea of Theorem \ref{thm:balanced_disconnected} to the case of non-balanced signed graphs. More precisely, cut-sets in the Gremban expansion can be related to frustration sets in the signed graph. Let us recall the definition of these two concepts:

\begin{definition}[Cut-set]  \label{def:cut_set}
    Let $G=(V,E,\sigma)$ be a (possibly unsigned) graph and $V=S\cup T$ a partition of its node set. The \emph{cut-set} $C(S,T)$ associated to this partition is the set of all edges with one end in each subset:
    \[
        C(S,T)
        \ =\ 
        \bigl\{(u,v)\in E \ :\  u\in S,\ v\in T\}\ .
    \]
    The \emph{edge connectivity} of $G$, denoted $\kappa_e(G)$, is the size of a smallest cut-set in the graph:
    \[
        \kappa_e(G)\ =\ \min_{(S,T)}\bigl| C(S,T)\bigr| ,\text{ taken~over all bipartitions $V=S\cup T$.}
    \]
\end{definition}
In the Gremban expansion, cut-sets inherit the symmetry of the partition: the cut-set associated with a Gremban-symmetric partition is itself Gremban-symmetric (Proposition \ref{prop:GS_cut_set} in Appendix \ref{app:extra_theorems}).

\begin{definition}[Frustration set]  \label{def:frustration_set}
    Let \(G=(V,E,\sigma)\) be a signed graph, and $\theta$ a switching function. 
    The \emph{frustration set} $F(\theta)$ associated to this switching function is the set of edges with negative sign after switching:
    $$
    F(\theta) = \{ uv\in E \,:\, \theta(u)\theta(v)\sigma(uv)=-1\} = \{uv\in E\,:\, \sigma_\theta(uv)=-1\}.
    $$
    The \emph{frustration index} of \(G\), denoted \(\varphi(G)\), is the size of the smallest frustration set in the graph:
    \[
        \varphi(G)
        \ =\ 
        \min_{\theta}
        \bigl|F(\theta)\bigr|, \text{~taken over all switching functions $\theta$.}
    \]
\end{definition}
The frustration set $F(\theta)$ measures the inconsistency or tension in the signed graph $G_\theta$, i.e., how far the signed graph is from being balanced. Removing a frustration set turns a graph into a balanced graph. Since balanced graphs always have zero frustration index, the frustration index is thus a measure of signed network imbalance. Similarly, the edge connectivity is a measure of how far a network is from being disconnected. Frustration sets in signed graphs are closely related to cut-sets in unsigned graphs: given a partition $(S,T)$ of an unsigned graph, then with the switching function $\theta(S)=+1$ and $\theta(T)=-1$, we have $F(\theta)=C(S,T)$. We can furthermore relate cut-sets and frustration sets in a signed graph to cut-sets in its Gremban expansion:

\begin{proposition}  \label{prop:frustration_cut}
Let \(G\) be a signed graph with edge connectivity $\kappa_e(G)$ and  frustration index \(\varphi(G)\), and let $\mathcal{G}$ be its Gremban expansion with edge connectivity $\kappa_e(\mathcal{G})$. 
Then 
\begin{align}  \label{eq:frustration_cut}
\kappa_e\bigl(\mathcal G\bigr)\le2\min\{\kappa_e(G),\varphi(G)\}.     
\end{align}
\end{proposition}
    
\begin{proof}
By definition of the frustration index, we can switch the signed graph \(G\) so that it has exactly \(\varphi\) negative edges.  In the corresponding expanded graph \(\mathcal G\), each negative edge of \(G\) becomes two edges joining the two copies of its endpoints, so there is a set of \(2\varphi(G)\) edges whose removal separates those two copies.  Hence \(\kappa_e\bigl(\mathcal G\bigr)\le2\varphi(G)\).
On the other hand, let \(C\) be a minimal edge cut in \(G\) of size \(\lvert C\rvert=\kappa_e(G)\).  In \(\mathcal G\), each edge of \(C\) corresponds to two edges, and deleting those \(2\kappa_e(G)\) edges disconnects \(\mathcal G\).  Therefore \(\kappa_e\bigl(\mathcal G\bigr)\le2\kappa_e(G)\). Combining these two bounds gives \(\kappa_e\bigl(\mathcal G\bigr)\le2\min\{\kappa_e(G),\varphi(G)\}\), as claimed.
\end{proof}
This is a first relation between cuts in the Gremban expansion and cuts and frustration sets in the underlying signed graphs. Taking into account symmetries, this turns into an even stronger relation:

\begin{theorem}\label{thm: cuts become cuts or frustration sets}
Let $G$ be a connected signed graph with Gremban expansion $\mathcal{G}$. If $V(\mathcal{G})=\mathcal{U}_1\cup\mathcal{U}_2$ is a Gremban-symmetric partition, then the cut-set $C(\mathcal{U}_1,\mathcal{U}_2)\subseteq E(\mathcal{G})$ is Gremban-symmetric, and
\begin{itemize}
    \item if $\eta(\mathcal{U}_1)=\mathcal{U}_2$ and $ \eta(\mathcal{U}_2)=\mathcal{U}_1$, then the projected cut-set $\pi(C(\mathcal{U}_1,\mathcal{U}_2))\subseteq E(G)$ is a frustration set;
    \item if $\eta(\mathcal{U}_1)=\mathcal{U}_1$ and $ \eta(\mathcal{U}_2)=\mathcal{U}_2$, then the projected cut-set $\pi(C(\mathcal{U}_1,\mathcal{U}_2))\subseteq E(G)$ is a cut-set;
\end{itemize}
Conversely, every cut-set and frustration set in $G$ lifts to a Gremban-symmetric cut-set in $\mathcal{G}$. This is a bijection.
\end{theorem}

\begin{proof}
Gremban-symmetry of the cut-set follows immediately from Gremban-symmetry of the partition, as $\eta(C(\mathcal{U}_1,\mathcal{U}_2))=C(\eta(\mathcal{U}_1),\eta(\mathcal{U}_2))=C(\mathcal{U}_1,\mathcal{U}_2)$. We start with the case $\eta(\mathcal{U}_1)=\mathcal{U}_2$ and $\eta(\mathcal{U}_2)=\mathcal{U}_1$. We will show that the cut-set is projected to the negative edges in $G_\theta$, with respect to the switching function $\theta(\{v:v^-\in\mathcal{U}_1\})=-1$. For every edge $(u^\chi,v^{\psi})\in C(\mathcal{U}_1,\mathcal{U}_2)$, i.e., with $u^\chi\in \mathcal{U}_1$ (and thus $v^\psi\in\mc U_2\text{~and~}v^{-\psi}\in \mc U_1$), we have by construction
$$
\sigma(uv)=\chi\cdot\psi\quad\text{~and~}\quad\theta(u)=\chi\quad\text{~and~}\quad\theta(v)=-\psi \quad\text{~and thus~}\quad \sigma_\theta(uv)=-1.
$$
The same construction shows that edges within $\mathcal{U}_1$ or $\mathcal{U}_2$ are positive in $G_\theta$ and thus that $\pi(C(\mathcal{U}_1,\mathcal{U}_2))$ is indeed a frustration set. When $\eta(\mathcal{U}_1)=\mathcal{U}_1$ and $\eta(\mathcal{U}_2)=\mathcal{U}_2$, the partition of $\mathcal{G}$ projects to a partition $V(G)=\pi(\mathcal{U}_1)\cup\pi(\mathcal{U}_2)$ of $G$, and the cut-set in $\mathcal{G}$ projects to the cut-set in $G$.

For the converse and bijectivity, we first observe that every cut-set $C(U_1,U_2)$ in $G$ lifts to a Gremban-symmetric cut-set $C(\{x^\pm:x\in U_1\},\{x^\pm:x\in U_2\})$ in $\mathcal{G}$, which projects back to the original cut-set. Second, let $F\subseteq E(G)$ be a frustration set of $G$. Since the subgraph $H=G\setminus F$ is balanced, it lifts to a disconnected Gremban-symmetric subgraph $\mathcal{H}\subseteq \mathcal{G}$. As a result, the complement $\mathcal{F}=\mathcal{G}\setminus\mathcal{H}$ is a Gremban-symmetric cut-set; this is precisely the lift of $F$ and it projects back to $F$.
\end{proof}
As a first example application of Theorem \ref{thm: cuts become cuts or frustration sets}, we can turn inequality \eqref{eq:frustration_cut} in Proposition \ref{prop:frustration_cut} into an equality, by restricting to a notion of edge connectivity for Gremban-symmetric cut-sets:
\begin{proposition}\label{prop:frustration_cut_cymmetric}
    Let \(G\) be a connected signed graph with frustration number \(\varphi(G)\) and edge-connectivity \(\kappa_e(G)\) and Gremban expansion $\mathcal{G}$. Furthermore, define the symmetric edge connectivity of $\mathcal{G}$ as
    \[
    \kappa_e^{\mathrm{sym}}\bigl(\mc G\bigr)
    \ =\ 
    \min\Bigl\{\lvert \mathcal{C}\rvert : \mathcal{C}\subseteq E\bigl(\mc G\bigr)\text{ is a Gremban-symmetric cut-set}\Bigr\}.
    \]
    Then
    \[
    \kappa_e^{\mathrm{sym}}\bigl(\mc G\bigr)
    =2\cdot \min\bigl\{\kappa_e(G),\ \varphi(G)\bigr\}.
    \]
\end{proposition}

\begin{proof}
This follows from the bijection between Gremban-symmetric cut-sets of $\mathcal{G}$ and frustration and cut-sets in $G$ in Theorem \ref{thm: cuts become cuts or frustration sets}, and the fact that every edge in the latter lifts to two copies in the former.
\end{proof}
%
\section{Algebraic formulation of the Gremban expansion}
\label{sec:gremban_algebraic}

So far, we have treated the Gremban expansion from a combinatorial point of view. In this section we reformulate the Gremban expansion in matrix terms. This algebraic perspective provides a compact representation of the graphs and their relations, and also incorporates the structural symmetries at play in a natural manner. 

We start by fixing some notation: we will work with vectors in and matrices acting on both $\R^n$ and $\R^{2n}$. Similar to the graph setting, to each basis vector $\mathbf{e}_{v}$ of $\R^n$ with $v\in\{1,\dots,n\}$, i.e., to each $v$th coordinate of a vector, we associate two basis vectors $\mathbf{e}_{v^+},\mathbf{e}_{v^-}$ of $\R^{2n}$. We write $(\mathbf{x}^+,\mathbf{x}^-)^\top\in\R^{2n}$ where $\mathbf{x}^+\in\R^n$ denotes the vector supported on the positive coordinates, and similar for $\mathbf{x}^-$. All matrices $M$ will be symmetric and thus have real eigenvalues, which are recorded as a multiset in the spectrum $\spec(M)$. Two matrices are similar if they differ by conjugation with a unitary matrix as $M=UNU^{-1}$; we write $M\sim N$.
%
\subsection{Gremban expansion of a matrix: definition and symmetries}
Let $M\in\R^{n\times n}$ be a real symmetric matrix with a decomposition into two symmetric matrices $M^+, M^-$ as follows

\begin{equation}\label{eq: matrix decomposition}
M=M^+-M^- \quad\text{~and~}\quad \bar{M} := M^+ + M^-.
\end{equation}
For the purpose of this article, we will think of $M^+$ as the submatrix containing all nonnegative entries of $M$, and $M^-$ containing all nonpositive entries, but for many results this is not necessary. Equation \eqref{eq: matrix decomposition} furthermore introduces the matrix $\bar{M}$, which is the `unsigned' version of $M$. We define the Gremban expansion of a matrix $M$ with an additive decomposition into two matrices:

\begin{definition}[Gremban matrix expansion] The \emph{Gremban expansion} of a real symmetric $n\times n$ matrix $M=M^+-M^-$ is the $2n\times 2n$ matrix $\mathcal{M}=\mathcal{M}(M)$ with block decomposition
\begin{align}  \label{eq:gremban_expansion}
    \mathcal{M} := \begin{pmatrix}
M^+ & M^-\\M^-& M^+
\end{pmatrix}.
\end{align}    
\end{definition}
The Gremban expansion of a matrix is an injective, non-surjective, nonlinear map (Proposition \ref{prop:injective_surjective_nonlinear} in Appendix \ref{app:extra_theorems}). Furthermore, similar to the Gremban expansion of graphs, the Gremban matrix expansion satisfies an involutive symmetry. Let $\1$ denote the identity matrix, then we define the $2n\times2n$ Gremban involution matrix $\mathcal{N}$ as
$$
\mathcal{N} := \begin{pmatrix}
    0&\1\\\1&0
\end{pmatrix}.
$$
As an operation on $\R^{2n}$, this matrix interchanges the positive and negative coordinates, as $\mathcal{N}(\mathbf{x}^+,\mathbf{x}^-)^\top=(\mathbf{x}^-,\mathbf{x}^+)^\top$, similar to the Gremban involution on $V(\mathcal{G})$. A $2n\times 2n$ matrix $\mathcal{M}$ is called \emph{Gremban-symmetric} if it is invariant under conjugation with the Gremban involution, as $\mathcal{M}=\mathcal{N}\mathcal{M}\mathcal{N}^{-1}$. We give a characterization:

\begin{proposition}\label{prop: characterization Gremban-symmetric matrix}
A symmetric $2n\times 2n$ matrix $\mathcal{M}$ is Gremban-symmetric if and only if it has block-form
$$
\mathcal{M} = \begin{pmatrix}
    A&B\\B&A
\end{pmatrix} \text{~for some symmetric matrices $A,B\in\R^{n\times n}$}.
$$
\end{proposition}

\begin{proof}
The claim follows directly by noting that the conjugation of $\left(\begin{smallmatrix}A&B\\C&D\end{smallmatrix}\right)$ by $\mathcal{N}$ equals $\left(\begin{smallmatrix}D&C\\B&A\end{smallmatrix}\right)$, hence it is invariant if and only if $A=D$ and $B=C$.
\end{proof}
Proposition \ref{prop: characterization Gremban-symmetric matrix} shows that the data $M^+,M^-$ and thus $\bar{M},M$ can be retrieved from a Gremban-symmetric matrix by looking at its matrix blocks. To recover the matrices $\bar{M}$ and $M$ algebraically from $\mathcal{M}$, we will make use of the following two projection operators:
\begin{equation}\label{eq: projection matrices}
    \Pi_s := \frac{1}{\sqrt{2}}\begin{pmatrix}\1 & \1\end{pmatrix} \in\R^{n\times 2n}\quad\text{~and~}\quad \Pi_a := \frac{1}{\sqrt{2}}\begin{pmatrix}\1 & -\1\end{pmatrix} \in\R^{n\times 2n}.
\end{equation}
The rows of these matrices together span $\R^{2n}$. Acting on vectors in $\R^{2n}$, the matrices $\Pi_s$ and $\Pi_a$ decompose the entries ${x}_v^+,{x}_v^-$ corresponding to the two signed entries of a coordinate $v$ into a symmetric part $({x}_v^++{x}_v^-)/\sqrt{2}$ and an antisymmetric part $({x}_v^+-{x}_v^-)/\sqrt{2}$, respectively. The square root appears because we are mainly interested in the conjugate action of $\Pi_s,\Pi_a$, which recovers the matrices $M$ and $\bar{M}$.
\begin{proposition}\label{prop: algebraic projections}
Let $M=M^+-M^-$ be a matrix and $\mathcal{M}$ its Gremban expansion. Then we find
$$
\Pi_s\mathcal{M}(M)\Pi_s^\top = \bar{M} \quad\text{~and~}\quad \Pi_a\mathcal{M}(M)\Pi_a^\top = M.
$$
\end{proposition}
%
\subsection{Spectral properties}
An important consequence of the construction of the Gremban expansion is that it produces a similar matrix to $\bar{M}\oplus M=\left(\begin{smallmatrix}
\bar{M}&0\\0&M
\end{smallmatrix}\right)$; recall that two matrices are similar if they differ by conjugation with a unitary matrix and that, as a consequence, these matrices share the same eigenvalues. Similarity of the Gremban expansion is most clearly observed when changing basis via the unitary matrix
\begin{equation}\label{eq:change_of_basis}
\mathcal{U} := \frac{1}{\sqrt{2}}\begin{pmatrix}
\1 & \1\\\1&-\1
\end{pmatrix} = \begin{pmatrix}
\Pi_s\\\Pi_a
\end{pmatrix}.
\end{equation}
We find the following similarity relation and spectral consequences for the Gremban expansion of a matrix:
\begin{theorem}\label{thm: spectral properties of Gremban symmetric matrix}
Let $\mathcal{M}$ be a Gremban-symmetric matrix with block-decomposition $\mathcal{M}=\left(\begin{smallmatrix}
M^+&M^-\\M^-&M^+
\end{smallmatrix}\right)$ and let $M=M^+-M^-$ and $\bar{M}=M^++M^-$. Then $\mathcal{M}$ is similar to $\bar{M}\oplus M$ and $\rm{Spec}(\mathcal{M})=\rm{Spec}(\bar{M})\cup\rm{Spec}(M)$ as multisets. Furthermore, every eigenpair $(\lambda,\mathbf{x})$ of $M$ lifts to an antisymmetric eigenpair $(\lambda, (\mathbf{x},-\mathbf{x})^\top)$ of $\mathcal{M}$, and every eigenpair $(\mu,\mathbf{y})$ of $\bar{M}$ lifts to a symmetric eigenpair $(\mu,(\mathbf{y},\mathbf{y})^\top)$ of $\mathcal{M}$. 
\end{theorem}

\begin{proof}
Direct computation shows that $\mathcal{U}\mathcal{M}\mathcal{U}^{-1}= \bar{M}\oplus M$. This confirms that the Gremban matrix $\mathcal{M}$ is similar to $\bar{M}\oplus M$ and thus, as a property of matrix similarity and block-diagonal matrices, that the spectrum of $\mc M$ is the union of the spectra of $\bar{M}$ and $M$. It remains to determine the eigenvectors of $\mathcal{M}$.

Following the eigenequation $M\mathbf{x}=\lambda\mathbf{x}$, antisymmetric lifts of eigenvectors of $M$ are eigenvectors of $\mathcal{M}$:
\[
\mc M\begin{pmatrix}\x\\-\x\end{pmatrix} = 
\begin{pmatrix}
    \1 & \1 \\ \1 & -\1 
\end{pmatrix} \begin{pmatrix}\bar M & 0\\0 & M\end{pmatrix} \begin{pmatrix}\0\\ \x\end{pmatrix} = 
\lambda \begin{pmatrix}
    \1 & \1 \\ \1 & -\1 
\end{pmatrix} \begin{pmatrix} \0\\ \x\end{pmatrix} = \lambda  \begin{pmatrix}\x\\ -\x\end{pmatrix}.
\]
A similar derivation shows that an eigenpair $(\mu,\mathbf{y})$ of $\bar{M}$ leads to an eigenpair $(\mu,(\mathbf{y},\mathbf{y})^\top)$ of $\mathcal{M}$, with symmetric eigenvector lift.

Hence every eigenvector of $\bar{M}$ or $M$ lifts to an eigenvector of $\mc M$. Finally, note that the symmetric and antisymmetric lifts of a pair of orthogonal bases of $\R^n$ combine to an orthogonal basis of $\R^{2n}$, and thus that the lifted eigenbases of $\bar{M},M$ combine to an eigenbasis of $\mathcal{M}$.
\end{proof}
%

\subsection{Adjacency and Laplacian matrix of the Gremban expansion}
As suggested by the names, the Gremban expansions of a graph and of a matrix are closely related. Indeed, for many natural matrix representations of a graph (see below), first passing to the Gremban expansion of the graph and then constructing the matrix representation gives the same result as first constructing the matrix representation of the graph and then passing to the Gremban expansion of the matrix. We discuss this in detail below for the adjacency and Laplacian matrices of a graph, but the same is true for other more specialized operators such as the normalized Laplacian matrix.
%
\subsubsection{Adjacency matrix} Let $G=(V,E,\sigma)$ be a signed graph. Its \emph{signed adjacency matrix} is the $n\times n$ symmetric matrix with entries $A_{uv}=A_{vu}=\sigma(u,v)$ for all $uv\in E(G)$ and with zeroes otherwise; this matrix is very well-studied both in graph theory and its applications. Writing $A^+=\max(A,0)$ for the adjacency matrix of the subgraph with only the positive edges, and $A^-=\max(-A,0)$ for {minus} the adjacency matrix of the subgraph with only negative edges, we can write the decomposition
$$
A = A^+-A^- \quad\text{~and~}\quad \bar{A} := A^++A^-,
$$
where $\bar{A}$ is called the \emph{unsigned adjacency matrix}. One directly confirms that the Gremban matrix expansion of the adjacency matrix $A$, with this decomposition, equals the adjacency matrix of the Gremban expansion $\mathcal{G}$. We call this the \emph{Gremban adjacency matrix} and write $\mathcal{A}$. Since $\mathcal{A}$ is Gremban-symmetric, Proposition \ref{prop: algebraic projections} and Theorem \ref{thm: spectral properties of Gremban symmetric matrix} apply. We mention some immediate corollaries (see also \cite{fox2017numerical}):
\begin{corollary}   
    The signed, unsigned and Gremban adjacency matrices $A,\bar{A}$ and $\mathcal{A}$ of a graph satisfy:
    $$
    \Pi_s\mathcal{A}\Pi_s^\top = \bar{A} \quad\text{~and~}\quad \Pi_a\mathcal{A}\Pi_a^\top = A.
    $$
\end{corollary}

\begin{corollary}
    Let $G$ be a balanced graph with adjacency matrix $A$ and Gremban adjacency matrix $\mathcal{A}$. Then $\mc A$ has the same eigenvalues as $A$, with double multiplicity: $\spec(\mc A) = \spec( A)\cup \spec(A)$.
\end{corollary}

\begin{proof}
    Since the graph is balanced, the adjacency matrix $A$ is similar to the unsigned adjacency matrix $\bar A$ via switching \cite{diaz2025mathematical,zaslavsky2013matrices}: $\bar A=D_\theta A D_\theta^{-1}$, where $D_\theta=\diag(\theta(1),\dots,\theta(n))$ is the diagonal matrix determined by the switching function that balances $G$. Since $D_\theta$ is unitary, we have $\bar{A}\sim A$ and $\mathcal{A}\sim A\oplus A$, as required.
\end{proof}

\begin{corollary} 
Let $G$ be a signed graph with adjacency matrix $A$ and unsigned adjacency matrix $\bar A$, and Gremban adjacency matrix \(\mc A \). Let \(\chi_M(\lambda) = \det(\lambda \1 - M)\) denote the characteristic polynomial of a matrix \(M\), \(\rho(M)\) its spectral radius, and \(\|\cdot\|_2\) the operator norm. Then:
\begin{align*}
\chi_{\mc A}(\lambda) = \chi_{\bar A}(\lambda) \chi_{A}(\lambda),\quad
\det(\mc A) = \det(\bar A)\det(A),\quad
\rho(\mc A) = \max\{\rho(\bar A), \rho(A)\}, \quad
\|\mc A\|_2 = \max\{\|\bar A\|_2, \|A\|_2\}.
\end{align*}
\end{corollary}

\begin{proof}
This follows directly from the similarity of $\mathcal{A}$ with the block-diagonal matrix $\bar{A}\oplus A$.
\end{proof}
Appendix~\ref{app:triangle_example} illustrates the algebraic Gremban expansion for the explicit example of a signed triangle.
%
\subsubsection{Laplacian matrix}
Let $G=(V,E,\sigma)$ be a signed graph, and define the \emph{degree} of a vertex $v\in V(G)$ as its number of neighbours $k(v):=\vert\{u:uv\in E(G)\}\vert=\sum_{u}\bar{A}_{uv}$;\footnote{The degree is understood in the combinatorial sense, disregarding edge signs; cf. alternative conventions in the literature like \cite{babul2024sheep}.} we write $K=\diag(k(1),\dots k(n))$ for the corresponding diagonal matrix and $\mathcal{K}:=K\oplus K$ for its Gremban expansion. The \emph{signed Laplacian matrix} $L$ and \emph{unsigned Laplacian matrix} $\bar L$ of $G$ are the $n\times n$ matrices, defined as \cite{hou2003laplacian,kunegis2010spectral}
$$
L = K - A^+ + A^- \quad\text{~and~}\quad \bar{L} = K - A^+ - A^-. 
$$
Note that for an unsigned graph (such as $\mathcal{G}$), both matrices coincide. As in the case of the adjacency matrix, one can directly confirm that the Laplacian matrix of the Gremban expansion equals the Gremban expansion of the Laplacian matrix, relative to the decomposition $L=(K-A^+) - (-A^-)$;\footnote{Note that this decomposition does not correspond to $\max(L,0)$ and $\max(-L,0)$.} more precisely, the \emph{Gremban Laplacian matrix} $\mathcal{L}$ equals
\begin{equation}\label{eq:expanded_laplacian}
\mathcal{L} = \begin{pmatrix}
K-A^{+} & -A^- \\ -A^- & K-A^+
\end{pmatrix} = \mathcal{K} - \mathcal{A}.
\end{equation}
Again, since $\mathcal{L}$ is a Gremban-symmetric matrix Proposition \ref{prop: algebraic projections} and the spectral theorem \ref{thm: spectral properties of Gremban symmetric matrix} apply. For the purpose of this article, we record the following specialization of the spectral Theorem to Gremban Laplacian matrices; this will be a key result in our spectral algorithm for community--faction detection. 

\begin{corollary}\label{cor:laplacian_spectrum_expansion}
Let $G$ be a signed graph. Its Gremban Laplacian $\mathcal{L}$ is similar to the matrix $\bar{L}\oplus L$ and thus $\spec(\mathcal{L})=\spec(\bar{L})\cup\spec(L)$. Furthermore, every eigenpair $(\lambda,\mathbf{x})$ of $L$ lifts to an antisymmetric eigenpair $(\lambda,(\mathbf{x},-\mathbf{x})^\top)$ of $\mathcal{L}$, and every eigenpair $(\mu,\mathbf{y})$ of $\bar{L}$ lifts to a symmetric eigenpair $(\mu,(\mathbf{y},\mathbf{y})^\top)$ of $\mathcal{L}$.
\end{corollary}
In particular, eigenvectors $\psi$ of $\mathcal{L}$ are either symmetric ($\mc N \psi=\psi$) or antisymmetric ($\mc N \psi=-\psi$) if $\spec(L)$ and $\spec(\bar{L})$ are disjoint; if the signed and unsigned Laplacian have shared eigenvalues, each eigenspace may still be decomposed into a symmetric and antisymmetric part. See Appendix~\ref{app:triangle_example} for an explicit example of the decomposition of Gremban Laplacian in the case of a signed triangle.

\begin{remark}
\label{rmk:involution_symmetry}
The symmetry--antisymmetry structure of the eigenvectors in Corollary~\ref{cor:laplacian_spectrum_expansion} can also be derived directly from Gremban-symmetry of the Gremban Laplacian: a direct computation shows that \( \mc N \mc L \mc N^{-1} = \mc L \), and thus \( \mc L \) commutes with \( \mc N \). This implies that $\mc N$ preserves the eigenspaces of $\mc L$, hence they must correspond to the symmetric vectors \( \mathbf{x}_s = (\mathbf{x},\mathbf{x})^\top \) and the antisymmetric ones \( \mathbf{x}_a =  (\mathbf{x},-\mathbf{x})^\top \).
\end{remark}
%
\section{Spectral detection of communities and factions}\label{sec:spectral_detection_faction_community}
In this section, we demonstrate how the Gremban expansion can be used to address the challenge of distinguishing communities and factions in signed networks---a problem that has received relatively little attention. Communities and factions represent two fundamentally different types of mesoscale organization in signed networks; yet existing methods often fail to separate them clearly, producing ambiguous or mixed partitions. Here, we leverage the spectral properties of the Gremban expansion to overcome this limitation and distinguish between communities and factions in a principled way. We begin by characterizing each structure.  A \emph{community} is a subset of nodes \( C \subseteq V \) whose internal connectivity is significantly larger than its connectivity to  \( V \setminus C \),  regardless of edge signs. In contrast, a \emph{faction} is a subset \( F \subseteq V \) for which most positive edges lie within $F$ or within its complement, and most negative edges cross between $F$ and $V\setminus F$. We use the term \emph{cluster} to refer to either a community or a faction.

\begin{figure}[h!]
    \centering
    \includegraphics[width=0.8\linewidth]{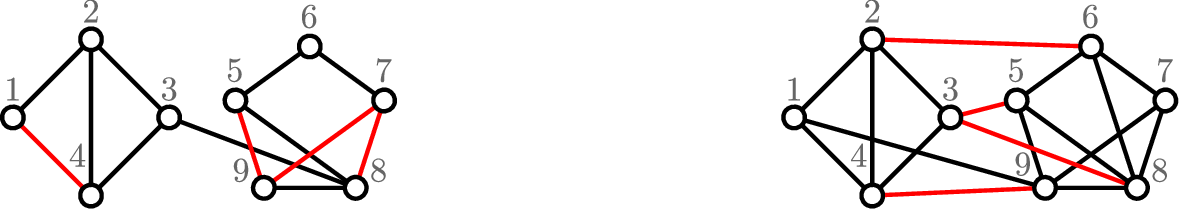}
\caption{Two types of mesoscale organization in a signed graph: community structure (left) and faction structure (right). In both cases, the optimal division of the node set is $\{1,2,3,4\},\{5,6,7,8,9\}$.}
\label{fig:communities_and_factions}
\end{figure}

Figure~\ref{fig:communities_and_factions} illustrates the two mesoscale structures. The signed graph on the left is organized into two communities, while the signed graph on the right shows a factional structure. Notice that the mesoscale organization on the left is independent of how the negative edges are distributed, i.e., it is a property of the underlying unsigned network. In contrast, in the signed graph on the right, the mesoscale structure only becomes apparent when taking into account the sign edges, particularly the fact that most edges between subsets 1-4 and 5-8 are negative.

From an optimization perspective, communities are partitions that minimize the size of the cut-set, while factions minimize the size of the frustration set. Proposition~\ref{prop:frustration_cut_cymmetric} shows that both objectives correspond to low-connectivity cuts in the Gremban-expanded graph. This observation allows us to treat both cases in 
a unified way: 
we expand the signed graph, compute
a Gremban-symmetric partition in the expanded graph, and then project the partition back to the original network.
A key algorithmic step in this approach is thus to cluster the Gremban-expanded graph. In this work, we use
spectral clustering \cite{ng2001spectral,shi2000normalized,vonLuxburg2007} due to its robustness, widespread use, and close connection to both cut-sets and spectral properties. Most importantly, as we show below, the proposed spectral methods naturally lead to Gremban-symmetric cuts which can be projected to the signed graphs.
Other approaches to clustering, such as \cite{peixoto2019bayesian,rosvall2008maps}, could be made compatible with this approach, by enforcing their output to respect the Gremban symmetry. 
%
\subsection{Projecting spectral cuts to communities and factions}
Spectral clustering refers to a class of methods that partition the nodes of a graph using the spectral properties of its Laplacian matrix. The central idea is to relax the combinatorial problem of minimizing the cut size between subsets into a continuous optimization problem involving eigenvectors \cite{vonLuxburg2007}. In the case of a bipartition, one considers the second smallest eigenvalue of the Laplacian---known as the \emph{algebraic connectivity}---and its associated eigenvector, the \emph{Fiedler vector}, denoted by $\psi_2$ \cite{fiedler1973algebraic}. The components of this vector provide a one-dimensional embedding of the nodes along a line, and a natural partition is obtained by thresholding its entries, for instance at zero. This yields an approximate solution to the normalized cut problem that penalizes small partitions while primarily minimizing inter-cluster connectivity. This spectral embedding captures global structure and thus typically outperforms purely local or greedy methods. Appendix~\ref{app:spectral_clustering} provides a detailed exposition of this approach.

We begin with the case of two clusters, where node membership is determined by the sign of the Fiedler eigenvector entries. We denote the eigenvalues and eigenvectors of $\mc L$ by $\lambda_j, \psi_j$ respectively, ordered such that $\lambda_j$ are non-decreasing. Corollary~\ref{cor:laplacian_spectrum_expansion} ensures that the Fiedler eigenvector $\psi_2$ of the expanded Laplacian 
$\mc L$ is either symmetric or antisymmetric. Remarkably, this dichotomy has a direct interpretation in the original signed graph: a symmetric eigenvector yields a Gremban-symmetric cut whose projection corresponds to a cut-set (i.e. a community partition), whereas an antisymmetric eigenvector yields a Gremban-symmetric cut whose projection corresponds to a frustration set, i.e., a factional partition.

\begin{theorem}
\label{thm:fiedler_cut_projection}
Let $G$ be a signed graph with Gremban expansion $\mathcal{G}$, let $\psi$ be an eigenvector of the Gremban Laplacian matrix $\mathcal{L}$ with nonzero entries, and define the bipartition $V(\mathcal{G})=\mathcal{U}_1\cup\mathcal{U}_2$ as
$$
\mathcal{U}_1=\{v\in V(\mathcal{G}):\psi(v)>0\}\quad\text{~and~}\quad\mathcal{U}_2=\{v\in V(\mathcal{G}):\psi(v)<0\}.
$$
\begin{itemize}
    \item If $\psi$ is symmetric, then the cut-set $C(\mathcal{U}_1,\mathcal{U}_2)$ is Gremban-symmetric and projects to a cut-set in $G$. 
    \item If $\psi$ is antisymmetric, then $C(\mathcal{U}_1,\mathcal{U}_2)$ is Gremban-symmetric and projects to a frustration set in $G$.
\end{itemize}
\end{theorem}

\begin{proof}
If $\psi$ is symmetric or antisymmetric we find that the cut-set $C(\mathcal{U}_1,\mathcal{U}_2)$ is Gremban-symmetric:
$$
(u,v)\in C(\mathcal{U}_1,\mathcal{U}_2) \iff \psi(u)\psi(v)<0 \overset{\text{(anti-)symm.}}{\iff} \psi(\eta(u))\psi(\eta(v))<0 \iff \eta(u,v)\in C(\mathcal{U}_1,\mathcal{U}_2).
$$
If $\psi$ is symmetric, then $\psi(v)>0$ implies $\psi(\eta(v))>0$ which means that $\eta$ fixes the sets $\mathcal{U}_1,\mathcal{U}_2$. If $\psi$ is antisymmetric, then $\psi(v)>0$ implies $\psi(\eta(v))<0$ which means that $\eta$ interchanges the sets $\mathcal{U}_1,\mathcal{U}_2$. The projection to cut-sets and frustration sets in $G$ then follows from Theorem \ref{thm: cuts become cuts or frustration sets}.
\end{proof}

\begin{remark}  \label{rem:psi_neq_0_case}
In Theorem \ref{thm:fiedler_cut_projection} we have assumed that the eigenvector $\psi$ has nonzero entries. 
When this vector has zero entries (i.e. $\psi(v^{\chi})=0$), then the corresponding nodes $v^\chi,v^{-\chi}$ can be assigned to either cluster without affecting the size of the cut-set. In such cases, we may assign \( v^\chi \) to $\mathcal{U}_1$, and $v^{-\chi}$ to $\mathcal{U}_1$ if $\psi$ is symmetric and to $\mathcal{U}_2$ if it is antisymmetric; this preserves the Gremban symmetry of the cut and guarantees that the result from the theorem still applies. This general case is treated in Theorem \ref{thm:generalized_fiedler_cut_projection} in Appendix \ref{app:extra_theorems}.
\end{remark}
Theorem~\ref{thm:fiedler_cut_projection} holds for any eigenvector of \( \mc L \) and, in fact, for any vector $\psi\in\R^{2n}$; however, our main interest lies in applying it to the Fiedler eigenvector \( \psi_2 \). This is because \( \psi_2 \) defines a cut-set $\mc C$ whose size closely approximates the edge connectivity of the expanded graph $\kappa_e(\mc G)$, as shown in Theorem~\ref{thm:min_cut_Laplacian} in Appendix \ref{app:spectral_clustering}. Since the cut-set is Gremban-symmetric, Proposition~\ref{prop:frustration_cut_cymmetric} implies that
\(
|\mc C| \approx \kappa_e(\mc G) = 2\min\left( \kappa_e(G), \varphi(G) \right),
\)
so that its projection satisfies
\(
\vert\pi(\mathcal{C})\vert \approx \min\left( \kappa_e(G), \varphi(G) \right).
\)
In other words, the projection of the partition of the expanded graph provides a good approximation to either the minimum cut partition or the minimum frustration partition, depending on which mesoscale structure dominates in the graph.

We now make explicit how projections of the Fiedler eigenvector of 
$\mc L$ distinguish between communities and factions. Let
$\mathbf{x}=(\mathbf{x^+},\mathbf{x^-})^\top=\operatorname{sign}(\psi_2)\in\{\pm1\}^{2n}$, and define the 
projections
\(
\mathbf{c}:=\Pi_s \mathbf{x}=\tfrac{1}{\sqrt{2}}(\mathbf{x}^++\mathbf{x}^-),\) and \(\mathbf{f}:=\Pi_a \mathbf{x}=\tfrac{1}{\sqrt{2}}(\mathbf{x}^+-\mathbf{x}^-).
\)
If $\psi_2$ is symmetric, then $\mathbf{x}^-=\mathbf{x}^+$, so $\mathbf{c}=\sqrt{2}\mathbf{x}^+$ and 
$\mathbf{f}=\mathbf{0}$, yielding a community partition. If $\psi_2$ is antisymmetric, 
then $\mathbf{x}^-=-\mathbf{x}^+$, so $\mathbf{f}=\sqrt{2}\mathbf{x}^+$ and $\mathbf{c}=\mathbf{0}$, yielding a factional 
partition. Thus the projection mechanism directly identifies the 
dominant mesoscale structure. When neither $\mathbf{c}$ nor $\mathbf{f}$ vanish, as can occur under eigenspace degeneracy, the symmetry is broken and 
further information is needed to determine the dominant structure (see 
Section~\ref{sec:Gremban_more_than_two_clusters}). Finally, note that this method is consistent with known limiting cases: in balanced connected graphs it recovers the balanced factions, while in unbalanced disconnected graphs it recovers the connected components (see Proposition~\ref{prop:clustering_balanced_connected} in Appendix~\ref{app:extra_theorems}).

We end this subsection summarizing the proposed procedure in the following algorithm:

\begin{algorithm}[H]
\caption{Detecting communities and factions via the Gremban Laplacian}
\label{alg:community_faction_detection}
\begin{algorithmic}[1]
\State \textbf{Input:} Signed graph \( G = (V,E,\sigma) \)
\State \textbf{Output:} Partition corresponding to dominant mesoscale structure
\State Compute the Laplacian matrix \( \mc L = \mc K - \mc A \) of $\mathcal{G}$ 
\State Compute the Fiedler eigenvector \( \psi_2 \) of \( \mc L \)
\State Let \( \mathbf{x} = \operatorname{sgn}(\psi_2) \)
\State Compute \( \mathbf{c} = \Pi_s \mathbf{x} \) and \( \mathbf{f} = \Pi_a \mathbf{x} \)
\If{\( \|\mathbf{c}\|_0 > \|\mathbf{f}\|_0 \)}
    \State Return the partition induced by \( \mathbf{c} \) (community structure) \,\,\quad\quad\text{~(see Theorem \ref{thm:fiedler_cut_projection})}
\ElsIf{\( \|\mathbf{f}\|_0 > \|\mathbf{c}\|_0 \)}
    \State Return the partition induced by \( \mathbf{f} \) (faction structure) \,\,\quad\quad\quad\quad\text{~(see Theorem \ref{thm:fiedler_cut_projection})}
\Else
    \State Network has more than two clusters; use Algorithm \ref{alg:community_faction_detection_multi} in Appendix \ref{app:extra_theorems}.
\EndIf
\end{algorithmic}
\end{algorithm}
%
\subsection{Numerical experiments}  \label{sec:numerical}

To test the proposed algorithm, we simulate signed networks using a degree-corrected stochastic block model with two ground-truth groups and $n=100$ nodes. Positive and negative edges are controlled by two separate block matrices, \(\rho^+\) and \(\rho^-\), each with parameters \(\rho^{\pm}_{\text{in}}\) and \(\rho^{\pm}_{\text{out}}\). We fix the positive parameters to \(\rho^+_{\text{in}} = 0.2\) and \(\rho^+_{\text{out}} = 0.02\), and interpolate from faction to community structure by varying \(\rho^-_{\text{in}}\) from 0 to 0.2 and letting \(\rho^-_{\text{out}} = 0.22 - \rho^-_{\text{in}}\). When \(\rho^-_{\text{in}} = 0\), most negative edges go across groups, inducing a faction structure. As \(\rho^-_{\text{in}}\) increases, the networks become less balanced and the inter-group edges become sparser, inducing a community structure. Further details on the network generation process can be found in Appendix \ref{app:ssbm}.

We compare three clustering algorithms applied to the simulated signed networks. The first method (\emph{Gremban method}), 
is our spectral method described in Algorithm \ref{alg:community_faction_detection}; it applies standard spectral clustering to the Gremban expansion of the graph and projects the results back onto the original nodes. The second method (\textit{signed method}) applies the same spectral algorithm directly to the signed Laplacian $L = K- A$, following \cite{kunegis2010spectral}. Finally, the third method (\textit{unsigned method}) applies the spectral algorithm using the unsigned Laplacian $\bar L$, which amounts to using the matrix \(|A|\) and thus ignoring edge signs. To account for degree heterogeneity, we use the normalized version of the Laplacian (see Appendix \ref{app:spectral_clustering}); however, our conclusions remain true for the unnormalized Laplacian.  For the Gremban and unsigned methods, we take the eigenvector associated to the second smallest eigenvalue, while for the signed method, we use the eigenvector associated to the smallest eigenvector, as $L$ does not in general have a constant eigenvector.

\begin{figure}
    \centering
    \includegraphics[width=\linewidth]{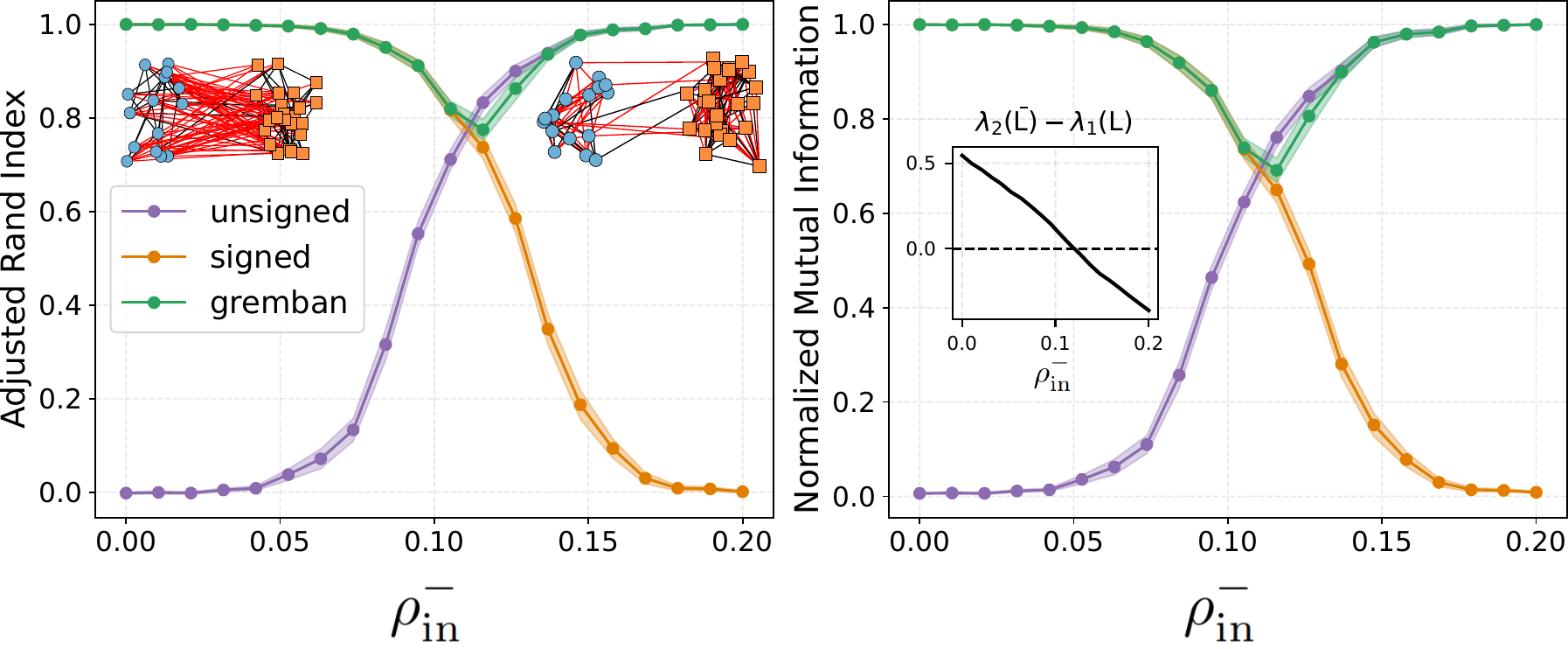}
    \caption{Clustering performance in a network of 100 nodes as a function of the negative in-group edge probability \(\rho^-_{\mathrm{in}}\), which interpolates between purely factional structure (\(\rho^-_{\mathrm{in}} = 0\)) and purely community structure (\(\rho^-_{\mathrm{in}} = 0.2\)). The negative out-group probability is set to \(\rho^-_{\mathrm{out}} = 0.22 - \rho^-_{\mathrm{in}}\), while the positive parameters are fixed at \(\rho^+_{\mathrm{in}} = 0.2\) and \(\rho^+_{\mathrm{out}} = 0.02\); see Appendix~\ref{app:ssbm} for full simulation details. We report the Adjusted Rand Index (left) and Normalized Mutual Information (right) for three spectral clustering methods: the unsigned Laplacian (`unsigned'), the signed Laplacian (`signed'), and the Gremban expansion (`Gremban'). The right panel includes an inset showing the spectral gap \(\lambda_2(\bar L) - \lambda_1(L)\) as a function of \(\rho^-_{\mathrm{in}}\). We also include example networks for low and high values of \(\rho^-_{\mathrm{in}}\); nodes are colored according to ground-truth group labels, positive edges appear as black lines, and negative edges as red lines. The shaded area represents the 95\% confidence interval, based on 100 independent runs per parameter setting.} 
    \label{fig:interpolation_communities_factions}
\end{figure}

We compared the detected clusters to the ground-truth ones using two metrics: the Adjusted Rand Index (ARI) and the Normalized Mutual Information (NMI) (see Appendix \ref{app:clustering_metrics} for details). The results are reported in Figure \ref{fig:interpolation_communities_factions}, which shows how clustering performance changes as the network structure shifts from factional (left) to community-dominated (right). As expected, the unsigned Laplacian recovers the ground-truth clusters when community structure dominates, while the signed Laplacian succeeds when faction structure is strongest. In contrast, the Gremban Laplacian accurately recovers both communities and factions, as evidenced by the high ARI and NMI values attained for both low and high values of $\rho_{in}^-$. Interestingly, no method perfectly retrieves the ground-truth labels in the range $\rho_{in}^-\in [0.1,0.15]$. In this parameter regime, the network transitions between community and faction structure, and thus lacks a clear mesoscale organization. This explanation is supported by the inset in Figure \ref{fig:interpolation_communities_factions}, which shows the gap between the second smallest eigenvalue of $\bar L$ and the smallest eigenvalue of $L$. Algorithm \ref{alg:community_faction_detection} selects the eigenvector lift associated to the smallest of the two eigenvalues, so the dominant mesoscale structure is determined by the sign of this spectral gap. We observe that the gap changes sign at $\rho_{in}^-\approx 0.14$, precisely where all methods exhibit poor performance. This confirms that the loss of clustering accuracy reflects a fundamental shift in the network's spectral structure.

Appendix~\ref{app:eigenvectors} examines in more details the eigenvalues and eigenvectors of the different Laplacian matrices $L,\bar{L}$ and $\mathcal{L}$ in networks with strongly pronounced factional or community structure. 
%
\subsection{Mixed community--faction structures and multi-way clustering}  \label{sec:Gremban_more_than_two_clusters}

    \begin{figure}[h]
        \centering
        \includegraphics[width=\linewidth]{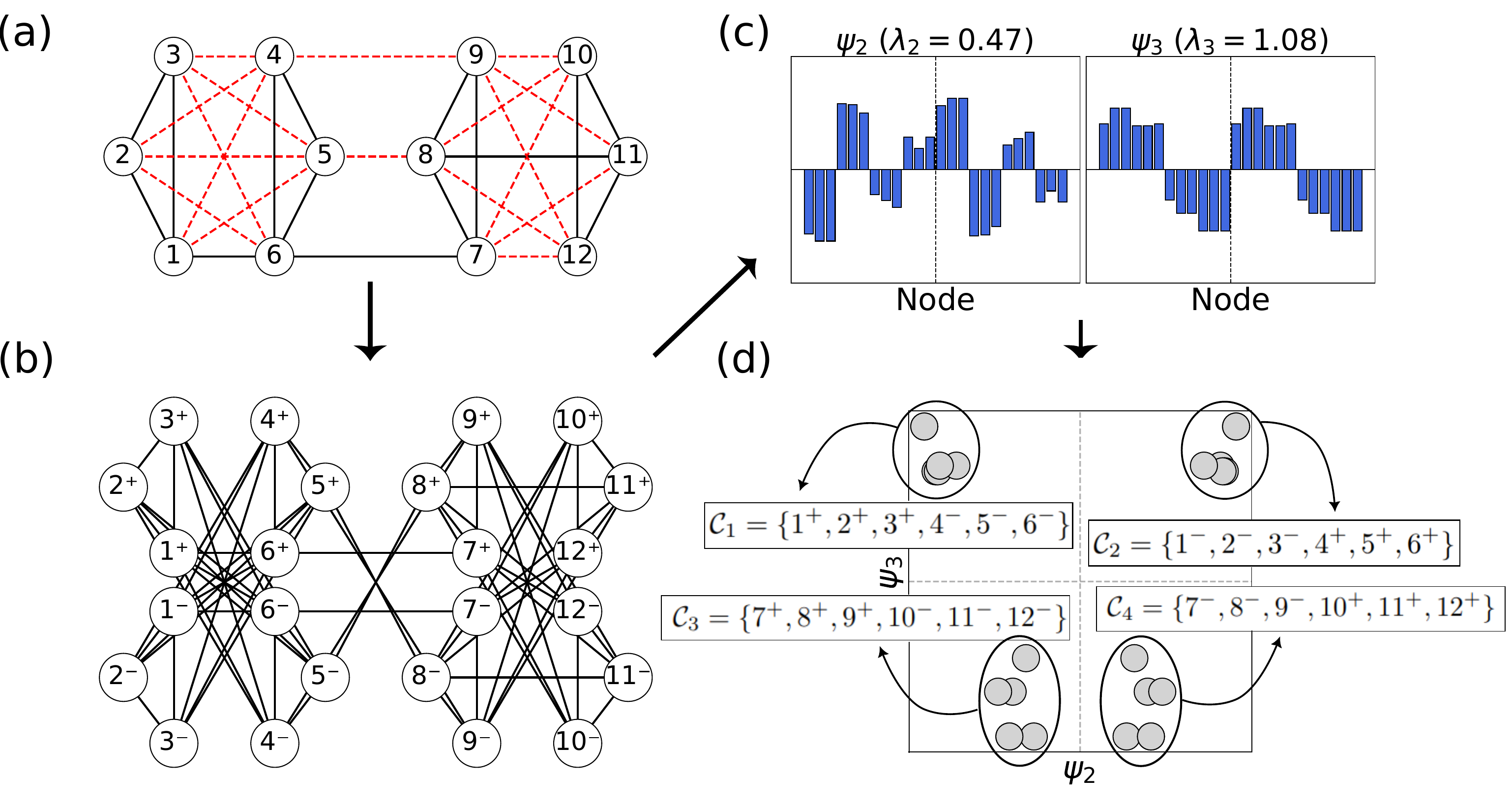}
        \caption{\textbf{Procedure to detect communities and factions coexisting in a network.} Given a signed network (panel~(a)), we first construct its Gremban expansion (panel~(b)) and compute the corresponding Laplacian \( \mathcal{L} \). We then extract the first two non-constant eigenvectors, \( \psi_2 \) and \( \psi_3 \) (panel~(c)), and use their components to embed each node into a two-dimensional space (panel~(d)). Applying a clustering algorithm in this embedded space yields four Gremban-symmetric clusters (panel~(d)). The clusters are then projected back to the original network using the one-sided projection \( \pi_+ \). The Gremban symmetry of the expanded clusters reveals which clusters correspond to antagonistic factions and which to communities. The eigenvalues associated with \(\psi_2\) and \(\psi_3\) are reported above the plots of panel~(c).
    }
    \label{fig:separating_communities_and_factions}
    \end{figure}

Up to this point we have focused on networks that exhibit either community or faction structure in isolation. Many real systems, however, contain both at once. A key advantage of the Gremban expansion that it naturally accommodates such cases, where communities and factions coexist in the same graph. When this happens, the expanded graph typically contains more than two clusters, requiring a multi-way extension of spectral clustering \cite{vonLuxburg2007}.

The standard multi-way clustering procedure, detailed in Appendix~\ref{app:spectral_clustering}, computes the Laplacian \( \mathcal{L} \) of the expanded graph \( \mathcal{G}(G) \) and extracts its leading \( k-1 \) non-trivial eigenvectors.  
These eigenvectors define a low-dimensional spectral embedding where each node $v^\chi$ is mapped to a point $\y_v^\chi = (\psi_2^\chi(v),\hdots, \psi_k^\chi(v))\in\R^{k-1}$. A clustering algorithm such as \( k \)-means is then applied to find clusters in these embedded points. 

As in the two-cluster setting, the key step is to interpret the clusters of the expanded graph in terms of the original signed network. This requires showing, first, that the clusters are Gremban-symmetric, and second, identifying whether their projection corresponds to a community or a faction.
\begin{proposition} \label{prop:gremban_distances}
    Let $\mc G$ be the Gremban expansion of a signed graph $G$ and $\y_v^\chi$ the embedding coordinates of its nodes. Then the Euclidean distance matrix 
    with entries $\mc D^{\chi \rho}_{vw} :=  \| \y_v^\chi - \y_w^\rho \|^2$ is Gremban-symmetric. 
\end{proposition}

\begin{proof}
    By Corollary \ref{cor:laplacian_spectrum_expansion}, each eigenvector satisfies either $\psi_j^\chi(v)=\psi_j^{-\chi}(v)$ or $\psi_j^\chi(v)=-\psi_j^{-\chi}(v)$, hence:
\begin{align*}
        \mc D^{\chi \rho}_{vw} &= \sum_{l=2}^k \left( \psi_l^{\chi}(v) - \psi_l^{\rho}(w) \right)^2 = \sum_{l=2}^k (\pm1)^2\left( \psi_l^{-\chi}(v) - \psi_l^{-\rho}(w) \right)^2 = \sum_{l=2}^k \left( \psi_l^{-\chi}(v) - \psi_l^{-\rho}(w) \right)^2 = \mc D^{-\chi, -\rho}_{vw}.
\end{align*}
\end{proof}
Proposition~\ref{prop:gremban_distances} shows that the spectral embedding preserves the involutive symmetry: the distance between two nodes equals the distance between their involutes. Therefore, any deterministic distance-based clustering method produces a Gremban-symmetric partition\footnote{In practice, algorithms such as k-means are usually not deterministic and depend on initialization, which can break the symmetry. Nevertheless, in our numerical experiments the resulting partitions were always Gremban-symmetric.}.

We now turn to the interpretation of the resulting clusters. We distinguish two cases according to their symmetry. If a cluster is invariant under the involution, $\eta(\mc U_i)=\mc U_i$, then the associated cut-set projects to a cut-set in $G$ (Theorem~\ref{thm: cuts become cuts or frustration sets}). Since the spectral clustering algorithm enforces low edge connectivity between $\mc U_i$ and its complement, the projection of $\mc U_i$ corresponds to a community (cf.\ the argument following Remark~\ref{rem:psi_neq_0_case}). If two clusters form an antisymmetric pair, $\eta(\mc U_i)=\mc U_j$ with $i\neq j$, the situation is subtler. The spectral embedding ensures that both $\mc U_i$ and $\mc U_j$ are well separated from the rest of the expanded graph. Their union $\mc U_i \cup \mc U_j$ is a Gremban-symmetric subset, so Theorem~\ref{thm: cuts become cuts or frustration sets} guarantees that the cut-set $C(\mc U_i \cup \mc U_j, V(\mc G)\setminus (\mc U_i \cup \mc U_j))$ projects to a cut-set in $G$ and the subgraph $\mc H$ induced by $\mc U_i \cup \mc U_j$ to a community. Furthermore, since $V(\mc H)=\mc U_i \cup \mc U_j$, the cut-set $C(\mc U_i,\mc U_j)$ projects to a frustration set of the subgraph $\pi(\mc H)$ with two opposing factions. Taken together, this shows that each pair of antisymmetric clusters in $\mc G$ project to a community containing nested factional divisions in $G$. The algorithm summarizing the proposed multi-way clustering method can be found in Appendix \ref{app:extra_theorems}.

To demonstrate our methodology, we apply it to a network containing both communities and factions (Figure~\ref{fig:separating_communities_and_factions}). For visualization purposes, we restrict the embedding to the first two nontrivial eigenvectors \( \psi_2 \) and \( \psi_3 \) of the Gremban Laplacian $\mc L$, which results in an embedding of the nodes into the plane. The panels in Figure~\ref{fig:separating_communities_and_factions} illustrate each step of the procedure: panel~(b) shows the Gremban expansion of the signed network in panel~(a); panel~(c) displays the eigenvectors \( \psi_2 \) and \( \psi_3 \); and panel~(d) presents the resulting two-dimensional spectral embedding. The embedding clearly reveals four Gremban-symmetric clusters: \( \mathcal{U}_1 = \{1^+, 2^+, 3^+, 4^-, 5^-, 6^-\} \), \( \mathcal{U}_2 = \{1^-, 2^-, 3^-, 4^+, 5^+, 6^+\} \), \( \mathcal{U}_3 = \{7^+, 8^+, 9^+, 10^-,11^-, 12^-\} \), and \( \mathcal{U}_4 = \{7^-, 8^-, 9^-, 10^+,11^+, 12^+\} \). The corresponding projections are \( U_1 = \pi_+(\mc U_1) = \{1, 2, 3\} \), \( U_2 = \pi_+(\mc U_2) = \{4,5,6\} \), \( U_3 = \{7,8,9\} \), and \( U_4 = \{10,11,12\} \). To determine whether these clusters represent communities or factions, we examine their Gremban symmetry. Since \( \eta(\mathcal{U}_1) = \mathcal{U}_2 \), we conclude that \( U_1 \) and \( U_2 \) form a pair of opposing factions. This is confirmed by inspecting the original network, where these groups are connected exclusively through negative edges. The same reasoning applies to \( U_3 \) and \( U_4 \), which also form a pair of opposing factions. Note that in this case the algorithm correctly handles the frustrated edge between nodes 8 and 11 and retrieves the faction structure, even though it is not perfectly balanced. To determine the communities, we note that no individual cluster is invariant under \( \eta \), but the unions \( \mathcal{U}_1 \cup \mathcal{U}_2 \) and \( \mathcal{U}_3 \cup \mathcal{U}_4 \) are. Consequently, \( U_1 \cup U_2 \) and \( U_3 \cup U_4 \) correspond to two communities in the original graph. Indeed, these groups are only loosely connected to each other, which aligns with our conceptualization of communities as sets of nodes with few links to the rest of the network.

In summary, we have shown that, as opposed to the spectral embeddings induced by $L$ or $\bar L$, the spectral embedding induced by $\mc L$ is capable of disentangling community from faction structure when both are present in a graph. This is a natural consequence of the fact that the spectral embedding induced by $\mc L$ combines eigenvectors from $L$ and $\bar L$ in a principled way. We can check this fact by noting that, in panel (c) of Figure \ref{fig:separating_communities_and_factions},  the eigenvector $\psi_2$ is antisymmetric and thus must come from $L$, while the eigenvector $\psi_3$ is symmetric and thus comes from $\bar L$. 
%
\section{Other applications of the Gremban expansion}\label{sec:other_applications}

Beyond its role in disentangling communities from factions, the Gremban expansion also provides a versatile framework for analyzing the structure and dynamics of signed networks. First, it enables a precise treatment of dynamics on signed networks by lifting them to an unsigned setting where standard tools apply. This construction essentially treats the signed dynamics as the net result of a positive and a negative process in the expanded graph. Here we examine the cases of random walks and diffusive dynamics, while more complicated dynamics are left for future work. These results also offer a complementary perspective on the interplay between communities and factions. In addition, the Gremban expansion provides a direct method to separately enumerate positive and negative walks. This combinatorial property, originally noted in \cite{fox2017numerical}, also facilitates the use of exponential and resolvent functions (see Appendix~\ref{app:walk_enumeration} for details).
%
\subsection{Signed random walks}

In a signed network, a random walker moves across edges as in the usual case. However, to model the effect of negative edges, walkers
are assigned a polarity that is flipped whenever they cross a negative edge \cite{tian2024spreading}. This process is governed by the equation $\x^{net}(t+1)=T\x^{net}(t)$, where $T =  K^{-1}A$ is the signed transition matrix. This equation only captures the net flow of walkers $\x^{net}$, which can be interpreted as the difference between the number of positive walkers $\mathbf{x}^+$ and that of negative walkers $\mathbf{x}^-$. As a result, the net number of walkers is not conserved, and in many cases the process converges to a trivial stationary state $\mathbf{x}^{net}\to0$. In contrast, the Gremban expansion arises as a natural tool to represent these dynamics in an expanded, unsigned space where polarity is encoded explicitly, allowing standard diffusion processes to capture the full walker distribution. 

To see how the Gremban expansion arises naturally, we expand the equation $\x^{net}(t+1)= T \x^{net}(t)$ \cite{tian2024spreading}:
\begin{align} \label{aux:random_walk}
    x_v^{net}(t+1)= \frac1{k(v)} \sum_w A^+_{vw} x_w^{net}(t) - A^-_{vw} x_w^{net}(t).
\end{align}
Grouping all positive and negative terms in Eq. \eqref{aux:random_walk} together
and substituting $\x^{net}=\x^+-\x^-$, we obtain:
\begin{align}  \label{aux:system_random_walk}
    x_v^+(t+1) = \frac1{k(v)} \sum_w A^+_{vw} x_w^+(t) + A^-_{vw} x_w^-(t), \qquad
    x_v^-(t+1) = \frac1{k(v)} \sum_w A^-_{vw} x_w^+(t) + A^+_{vw} x_w^-(t).
\end{align}
Equation \eqref{aux:system_random_walk} can be compactly written as:
\begin{align}  \label{eq:Gremban_random_walk}
    \x(t+1)=\mc T \x(t),  \textrm{  where  }
\x = \begin{pmatrix}
    \x^+ \\ \x^-
\end{pmatrix}  \textrm{  and  } \mc T = \mc K^{-1} \mc A = \begin{pmatrix}
K^{-1}A^+ & K^{-1}A^- \\
K^{-1}A^- & K^{-1}A^+
\end{pmatrix}.
\end{align}
We have thus found that the matrix governing the random walk evolution is the Gremban expansion of the transition matrix. Note that we did not impose the Gremban symmetry; instead, this symmetry follows from the fact that the observable $\x$ is a net count of walkers. Equation \eqref{aux:random_walk} can be easily recovered by multiplying Equation \eqref{eq:Gremban_random_walk} by the projection $\sqrt{2}\Pi_a$ (recall Eq. \eqref{eq: projection matrices}). Moreover, we can now obtain an equation for the total number of walkers $\x^{tot}:=\x^++\x^-$ if we multiply Equation \eqref{eq:Gremban_random_walk} by the projection $\sqrt{2}\Pi_s$:
\begin{align}  \label{aux:total_random_walk}
    \x^{tot}(t+1)= \sqrt{2} \Pi_s \mc T \x(t) = K^{-1} (A^+\x^++A^-\x^++A^-\x^-+A^+\x^-) = K^{-1}\bar A \x^{tot}(t),
\end{align}
The matrix \( K^{-1}\bar A \) defines the transition probabilities of a random walk on the underlying unsigned graph, so Equation~\eqref{aux:total_random_walk} describes a standard random walk over that structure.

The spectral properties of the matrix $\mathcal{T}$ allow us to relate the dynamics on the expanded graph to those of the original signed network. Since $\mathcal{T}$ is Gremban-symmetric, we know by Theorem \ref{thm: spectral properties of Gremban symmetric matrix} that it admits a block-diagonalization via the change of basis \eqref{eq:change_of_basis}, yielding
\(
\mc U \mc T \mc U^\top = \bar T \oplus T, \ \text{where } \bar T = K^{-1} \bar A.
\)
 Hence the spectrum of \( \mc T \) consists of the union of the spectra of $T$ and $\bar T$. Moreover, each eigenvector of \( \mc T \) is either a symmetric lifting of an eigenvector of \( \bar T \) or an antisymmetric lifting of an eigenvector of \( T \). In particular, the right-eigenvector corresponding to the unit eigenvalue is the all-ones vector $\mathbf{1}_{2n}$, and the left-eigenvector is $(\mathbf{k},\mathbf{k})^\top$, exactly the symmetric lifting of the Perron eigenpair of the unsigned walk $\bar T$. However, since $\mc T \sim \bar T \oplus T$, the unit eigenvalue can have multiplicity two if both $T$ and $\bar T$ admit it, which occurs precisely when the signed graph is structurally balanced:

\begin{proposition}
Let $G$ be a connected signed graph with Gremban-expanded transition matrix $\mc T$. Then, $G$ is balanced if and only if the eigenspace of $\mc T$ corresponding to the eigenvalue $\lambda=1$ is two-dimensional, spanned by the constant vector $\mathbf1_{2n}$ and the antisymmetric lift $(\vartheta,\,-\vartheta)^\top$, where $\vartheta:=(\theta(1),...\theta(n))$ and $\theta$ is the switching function that balances $G$. 
\end{proposition}

\begin{proof}
 Suppose \(G\) is balanced. 
Then there exists a switching matrix 
\(D_\theta=\diag(\theta(1),\dots,\theta(n))\) such that 
\(A=D_\theta \bar A D_\theta^{-1}\), and hence 
\(T=D_\theta \bar T D_\theta^{-1}\sim \bar T\). 
Since \(\bar T\), the transition matrix of a connected unsigned graph, 
has a simple unit eigenvalue with eigenvector \(\mathbf{1}_n\), 
it follows that \(T\) has a corresponding eigenvector 
\(D_\theta \mathbf{1}_n = \vartheta\). 
By Theorem~\ref{thm: spectral properties of Gremban symmetric matrix}, 
the antisymmetric lift \((\vartheta,-\vartheta)^\top\) is then an 
eigenvector of \(\mc T\) with eigenvalue \(1\), 
linearly independent of \(\mathbf{1}_{2n}\). 
 Conversely, suppose \(G\) is unbalanced. 
Then \(\mc G\) is connected 
(Theorem~\ref{thm:balanced_disconnected}), 
so \(\mc T\) is irreducible and entrywise nonnegative. 
By the Perron-Frobenius theorem, the eigenvalue \(1\) is simple and hence its associated eigenspace is one-dimensional.
\end{proof}
To see how this reflects in the original signed dynamics, we project the eigenvectors via $\Pi_s$ and $\Pi_a$.
Projecting $\mathbf1_{2n}$ gives
\(
\Pi_a\,\mathbf1_{2n}=\0,
\
\Pi_s\,\mathbf1_{2n}\propto\mathbf1_n
\). This means that in the unbalanced case, the signed random walk $\x(t+1)=T\x(t)$ has a trivial steady state, where every node has an equal amount of positive and negative walkers, while the unsigned dynamics (governed by $\bar T$) exhibits a uniform steady state. However, in the balanced case, the antisymmetric lift $(\vartheta, -\vartheta)^\top$ survives when projected via $\Pi_a$, yielding a polarized steady state where one of the two balanced factions is net-positive, while the other is net-negative. A balanced signed graph thus results in a reducible random walk.

Finally, we note that random walks on signed graphs also appear in context of random walks on simplicial complexes. As explained in the recent review \cite{eidi2023irreducibilitymarkovchainssimplicial}, in that context, reducibility of the random walk is in correspondence with orientability of the underlying simplicial complex. This is the same phenomenon as captured by the topological result that the Gremban expansion, i.e., the double cover, of a signed graph is disconnected if and only if it is balanced (Theorem \ref{thm:balanced_disconnected}).
%
\subsection{Diffusion dynamics}
The Gremban expansion can also disentangle signed continuous-time diffusive processes. Recall that diffusion on a signed graph is governed by the equation $\dot \x^{net} = -L\x^{net}$, where $L=K-A$ is the signed Laplacian \cite{altafini2012consensus}. The corresponding expanded dynamics are $\dot \x = -\mc L \x$, where $\mc L$ is the Gremban Laplacian matrix \eqref{eq:expanded_laplacian}. Since \(\mc L\) is symmetric and positive semidefinite, \(\dot \x=-\mc L \x\) converges to $\x_\infty \in \ker\mc L$. From Corollary \ref{cor:laplacian_spectrum_expansion}, we can deduce that the null space of $\mc L$ is the union of the symmetric lifting of $\ker(\bar L)$ and the antisymmetric lifting of $\ker(L)$.
Because \(\bar L\) always contains \(\mathbf1_n\) in its kernel, \(\mathbf1_{2n}\) is an eigenvector of \(\mc L\) with zero eigenvalue.  Moreover, if \(G\) is balanced then \(L\) has a zero eigenvector \(\vartheta\), whose antisymmetric lift \( (\vartheta,-\vartheta)^\top\) is a zero eigenvector of $\mc L$ that is linearly independent of \(\mathbf1_{2n}\). Projecting back via \(\Pi_s,\Pi_a\) shows \(\Pi_a\mathbf1_{2n}=\0\), \(\Pi_a(\vartheta,-\vartheta)^\top \propto\vartheta\), \(\Pi_s\mathbf1_{2n}\propto\mathbf1_n\), and \(\Pi_s(\vartheta,-\vartheta)^\top=\0\). In short, just as with the random walk, the Gremban expansion shows that the stationary dynamics of signed diffusion are trivially zero when the graph is unbalanced, but converge to a factional dissensus state when the graph is balanced. This recovers Altafini's result \cite{altafini2012dynamics}, but additionally reveals that the nontrivial stationary state stems from an additional antisymmetric mode in the expanded Laplacian.

\subsection{Communities and factions, revisited}
\begin{figure}
    \centering
    \includegraphics[width=\linewidth]{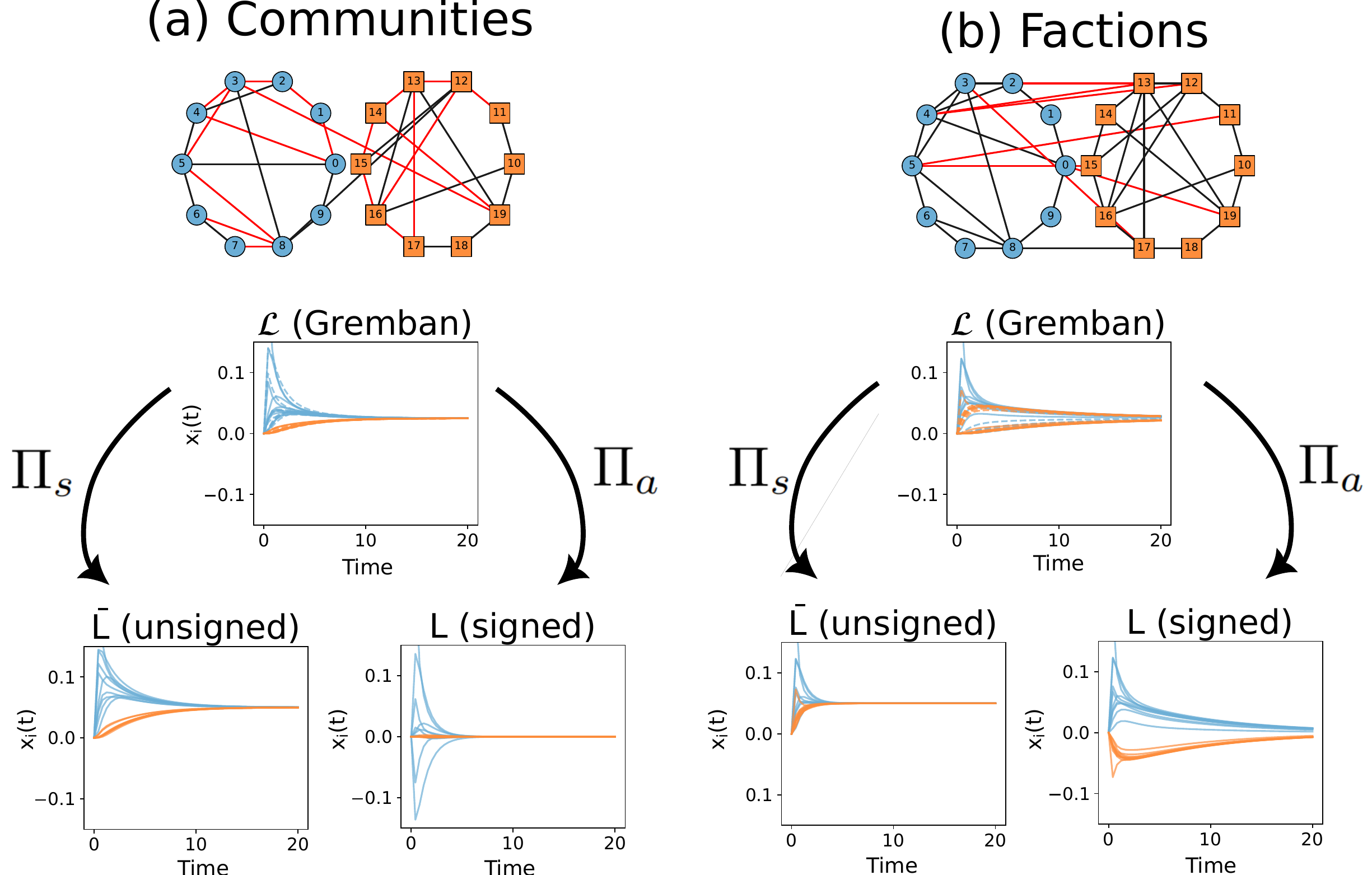}
    \caption{
\textbf{Diffusion dynamics reveal mesoscale structure.}
The upper row shows two synthetic signed networks with 20 nodes, divided into two equally sized groups (blue and orange). The left network has community structure while the right one displays factional structure. Black lines represent positive edges while red lines indicate negative edges. The middle row presents the time evolution of the state $x_v$ of every node in the expanded graph under the dynamics $\dot x = - \mc L x$ and initial condition $x_v(0) = \delta_{v0}$. Trajectories are colored blue or orange depending on the group to which the corresponding node belongs. Solid lines represent the positive-polarity component of each node, while dashed lines represent the negative-polarity component. The bottom row shows instead the signed and unsigned dynamics governed by $\dot \x = - L \x$ and $\dot \x = - \bar L \x$, respectively. Note that only the signed dynamics allow for negative states.
}
\label{fig:dynamics_and_mesoscale_structure}
\end{figure}

We now revisit the algorithm for detecting communities and factions in light of the diffusion dynamics induced by the Laplacian. The dynamics on the Gremban expanded graph, governed by \( \dot{\mathbf{x}} = -\mc L \mathbf{x} \), encompass both the signed dynamics \( \dot{\mathbf{x}} = -L \mathbf{x} \) and the unsigned dynamics \( \dot{\mathbf{x}} = -\bar L \mathbf{x} \). It is well established that diffusive processes can reveal mesoscale structures \cite{lambiotte2021modularity}, as such structures trap diffusive flows and produce long-lived metastable states, which in turn serve as hallmarks of the underlying organization (see Remark \ref{rem:diffusion} in Appendix \ref{app:spectral_clustering}). The same principle holds in Gremban-expanded graphs, where the metastable states reflect both community-like and factional structures.

To illustrate this, in Figure \ref{fig:dynamics_and_mesoscale_structure} we constructed two synthetic networks with 20 nodes, divided into two equally-sized groups. In panel (a), we show a graph where the two groups are sparsely connected-only two edges link them-and edge signs are assigned randomly, with a 60\% probability of being positive. This setup results in a {community-like} structure. The dynamics in the Gremban-expanded graph result in two long-lived metastable states that eventually converge to the stationary state. Pairs of nodes with opposite
polarities converge to the same metastable state, indicating the presence of communities. This is further confirmed by the projected dynamics: the unsigned dynamics exhibit a metastable plateau, while the signed dynamics decay quickly to the trivial steady state $\mathbf{x}\rightarrow\mathbf{0}$. 

In panel (b) we show a similar network, but with significantly more inter-group connections. Now, edge signs follow a structured pattern: all intra-group edges are positive, while inter-group edges are negative except for one edge, added to unbalance the graph and force convergence to $\mathbf{x}\rightarrow\mathbf{0}$.
In this case, the Gremban-expanded dynamics again exhibit a metastable state. However, in this case the nodes with positive polarity in one group converge to the same trajectory as those with negative polarity in the other group. This pattern reflects a factional structure. The corresponding unsigned dynamics hence decays rapidly to zero, while the signed one now exhibits a long-lived metastable state.

In summary, the Gremban-expanded dynamics unify the signed and unsigned diffusive processes and capture both community- and faction-like structures. While here we have focused on diffusion, the same principle can be extended to more complex dynamical processes, such as biased random walks or synchronization dynamics, to uncover more subtle forms of structure.
%
\section{Conclusion}\label{sec:conclusion}
We have presented a unified framework for analyzing signed networks by leveraging the Gremban expansion, a lifting technique that transforms signed graphs into unsigned ones while preserving their structural and spectral properties. Our main contribution lies in establishing a principled correspondence between 
symmetries of groups of nodes and edges in the expanded graph and community and factional structures 
in the original signed network. This insight allows us to reinterpret signed spectral clustering through the lens of symmetry classes: symmetric eigenmodes encode community structure, while antisymmetric modes reveal factional divisions.

Our framework provides a theoretical foundation that opens the door to the use of standard spectral clustering algorithms  for structure detection in signed networks.  Our approach explicitly distinguishes between frustration-based and density-based mesoscale patterns, bridging a methodological gap in the existing literature, which has focused on structurally balanced settings or treated signed edges as minor perturbations.

 Future work could explore how the
Gremban expansion reflects structural features beyond balance and mesoscale organization. For instance, the relationship between node centrality in the expanded graph and in the original signed network remains unclear. Moreover, as shown in this work, the Gremban expansion arises naturally in the study of dynamical processes on signed graphs. This suggests that the framework could support the extension of nonlinear dynamics (such as synchronization) to networks with signed interactions. In turn, this may offer a new perspective for uncovering hidden topological features of signed networks through their dynamics. 
Finally, since networks with complex-weighted edges can be viewed as a natural generalization of signed networks \cite{BottcherPorter2024,tian2024structural}, a promising direction for future research is to explore Gremban expansions within this broader setting.
In short, we believe this framework lays the groundwork for a more comprehensive theory of signed network analysis,  enabling the use of classical tools  while properly accounting for antagonistic interactions.

\section*{Acknowledgements}
FDD thanks funding from Program for Units of Excellences Maria de Maeztu (CEX2021-001164-M/10.13039/501100011033) and project APASOS (PID2021-122256NB-C22).
RL acknowledges support from the EPSRC Grants EP/V013068/1, EP/V03474X/1, and EP/Y028872/1.

\bibliographystyle{plain}
\bibliography{bib} 

\newpage
\appendix
%
\section{Additional theorems and algorithms} \label{app:extra_theorems}

\begin{proposition}  \label{prop:GS_cut_set}
    Let $\mc G$ be the Gremban expansion of a signed graph and $V(\mc G)=\mc U_1 \cup \mc U_2 $ a Gremban-symmetric partition of the node set. Then the corresponding cut-set $C(\mc U_1, \mc U_2)$ is Gremban-symmetric. 
\end{proposition}

\begin{proof}
Let $(u^\chi,v^\psi)\in C(\mc U_1,\mc U_2)$ be an edge with $u^\chi\in\mathcal{U}_1,u^\psi\in\mathcal{U}_2$. Since the partition is Gremban-symmetric, either $\eta(\mc U_1)=\mc U_1,\ \eta(\mc U_2)=\mc U_2$ or $\eta(\mc U_1)=\mc U_2,\ \eta(\mc U_2)=\mc U_1$. In the first case $u^{-\chi}\in\mc U_1$, $v^{-\psi}\in\mc U_2$ and in the second, $u^{-\chi}\in\mc U_2$, $v^{-\psi}\in\mc U_1$. Hence, in both cases $\eta(u^\chi,v^\psi)\in C(\mathcal{U}_1,\mathcal{U}_2)$ and thus $\eta(C(\mathcal{U}_1,\mathcal{U}_2))\subseteq C(\mathcal{U}_1,\mathcal{U}_2)$. Since $\eta^2=\mathrm{id}$, it follows that $\eta$ fixes the cut-set and thus that $C(\mathcal{U}_1,\mathcal{U}_2)$ is Gremban-symmetric.
\end{proof}

\begin{proposition}  \label{prop:injective_surjective_nonlinear}
The Gremban expansion \(\mc M \colon \R^{n\times n} \to \R^{2n \times 2n}\) defined by Equation~\eqref{eq:gremban_expansion} is injective, non-linear, and not surjective.
\end{proposition}

\begin{proof}
\emph{(1) Injectivity.} 
Suppose \(\mc M(M_1) = \mc M(M_2)\). Then,
\( \left(
\begin{smallmatrix}
M_1^+ & M_1^- \\
M_1^- & M_1^+
\end{smallmatrix} \right)
= \left(
\begin{smallmatrix}
M_2^+ & M_2^- \\
M_2^- & M_2^+
\end{smallmatrix}\right).
\)
It follows that \(M_1^+ = M_2^+\) and \(M_1^- = M_2^-\), so \(M_1 = M_1^+ - M_1^- = M_2^+ - M_2^- = M_2\). Thus, \(\mc M\) is injective.

\emph{(2) Nonlinearity.} Set $M_1=1, M_2 = -1$, and define $M_i^+:=\max(M_i,0), M_i^-:=\max(-M_i,0)$.
Then \(M_1 + M_2 = 0\), so \(\mc M(M_1+M_2) = \mc M(0) = 0\). On the other hand,
\(
\mc M(M_1) = \left(\begin{smallmatrix} 1 & 0 \\ 0 & 1 \end{smallmatrix} \right), 
\mc M(M_2) = \left(\begin{smallmatrix} 0 & 1\\ 1 & 0 \end{smallmatrix} \right),
\)
and hence
\[
\mc M(M_1) + \mc M(M_2) = 
\begin{pmatrix}
1 & 1 \\
1 & 1
\end{pmatrix}
\neq  \mc M(M_1 + M_2).
\]

\emph{(3) Non-surjectivity.}  Consider the matrix \( \mc M = \left(\begin{smallmatrix} 1 & 0 \\ 0 & 0 \end{smallmatrix} \right) \), and suppose \(\mc M = \mc M(M)\) for some \(M \in \R^{n \times n}\). This would require \(M^+ = 1 \) and \( M^+ = 0\) simultaneously, a contradiction. Thus, $\mc M$ is not surjective. 
\end{proof}

\begin{theorem}
\label{thm:generalized_fiedler_cut_projection}
Let $G$ be a signed graph with Gremban expansion $\mathcal{G}$, let $\psi$ be an eigenvector of the Gremban Laplacian matrix $\mathcal{L}$, possibly with entries equal to zero, and define the bipartition $V(\mathcal{G})=\mathcal{U}_1\cup\mathcal{U}_2$ as follows:
\begin{align*}
   &\mc U_1 := \{ v^\chi \in V(\mc G) : \psi^\chi(v) \geq 0 \}, \ \ \mc U_2 := \{ v^\chi \in V(\mc G) : \psi^\chi(v) < 0 \} \ \text{if $\psi$ symmetric}. \\
   &\mc U_1 := \{ v^\chi \in V(\mc G) : \psi^\chi(v) > 0 \} \cup \{ v^+ : \psi^+(v) = 0 \}, \ \  \mc U_2 := \{ v^\chi : \psi^\chi(v) < 0 \} \cup \{ v^- : \psi^-(v) = 0 \} \ \text{if $\psi$ antisymm}. 
\end{align*}
\begin{itemize}
    \item If $\psi$ is symmetric, then the cut-set $C(\mathcal{U}_1,\mathcal{U}_2)$ is Gremban-symmetric and projects to a cut-set in $G$. 
    \item If $\psi$ is antisymmetric, then $C(\mathcal{U}_1,\mathcal{U}_2)$ is Gremban-symmetric and projects to a frustration set in $G$.
\end{itemize}
\end{theorem}

\begin{proof}
If $\psi$ is symmetric, then $\psi^{\chi}(v)=\psi^{-\chi}(v)$, so
\[\eta(\mc U_1)=\{v^\chi:\psi^{-\chi}(v)\geq 0\} = \{v^\chi:\psi^{\chi}(v)\geq 0\} = \mc U_1.\] 
If $\psi$ is antisymmetric, then $\psi^{\chi}(v)=-\psi^{-\chi}(v)$, so
\[\eta(\mc U_1)=\left\{ v^\chi : \psi^{-\chi}(v) >
   0 \right\} \cup \left\{ v^-  : \psi^+(v) = 0 \right\}=\left\{ v^\chi  : \psi^\chi(v) < 0 \right\} \cup \left\{ v^-  : \psi^-(v) = 0 \right\}=\mc U_2.\]
Thus in both cases the partition $\mc U_1\cup\mc U_2$ is Gremban-symmetric. By Proposition~\ref{prop:GS_cut_set}, the induced cut-set $C(\mc U_1,\mc U_2)$ is also Gremban-symmetric, and Theorem~\ref{thm: cuts become cuts or frustration sets} identifies its projection in $G$ as a cut-set in the symmetric case and a frustration set in the antisymmetric case.
\end{proof}

\begin{proposition} \label{prop:clustering_balanced_connected}
Let $G$ be a signed graph, \(\mc L\) its Gremban-expanded Laplacian, $\psi_2$ its Fiedler eigenvector and $\mathbf{c}=\Pi_s\,\rm{sign}(\psi_2)$ and $\mathbf{f}=\Pi_a\,\rm{sign}(\psi_2)$. Then:
\begin{enumerate}
  \item If \(G\) is balanced and connected, then \(\psi_2\) is antisymmetric, \(\mathbf{c} = \0\), and \(\{\,v : f(v) = +\sqrt{2}\}\), \(\{\,v :  f(v) = -\sqrt{2}\}\) are the balanced factions.
  \item If \(G\) is unbalanced and disconnected, then \(\psi_2\) is symmetric, \(\mathbf{f} = \0\), and \(\{\,v :  c(v) = +\sqrt{2}\}\), \(\{\,v :  c(v) = -\sqrt{2}\}\) are the connected components.
\end{enumerate}
\end{proposition}

\begin{proof}
If $G$ is balanced and connected, the signed Laplacian \( L \) admits an eigenvector \(\vartheta \in \{\pm1\}^n \) associated with the zero eigenvalue, where the signs of each entry indicate the balanced faction of the corresponding node~\cite{kunegis2010spectral}. Its antisymmetric lift, \( (\vartheta, -\vartheta)^\top \), is an eigenvector of \( \mc L \) with eigenvalue zero. Since \( G \) is connected, the only other zero eigenvector of \( \mc L \) is the constant vector \( \mathbf{1}_{2n} \). Thus, \( (\vartheta, -\vartheta)^\top \) is the Fiedler eigenvector. When projected, it yields \( \c = \0 \) and \( \f = \sqrt{2}\vartheta \), confirming that the dominant mesoscale structure is factional. Moreover, the signs of $\f$ indicate the balanced factions of $G$.

Now suppose $G$ is unbalanced and has two connected components, \( U_1 \) and \( U_2 \). Then the indicator vectors \( \mathbf{1}_{U_1} \) and \( \mathbf{1}_{U_2} \) span the kernel of \( \bar L \), which has dimension two. From these, we can build an orthonormal basis for the kernel containing the normalized constant vector \( \mathbf{1}_n / \sqrt{n} \). The remaining vector is
\(
\mathbf{u} = \sqrt{\frac{n_2}{n_1 n}} \, \mathbf{1}_{U_1} - \sqrt{\frac{n_1}{n_2 n}} \, \mathbf{1}_{U_2},
\)
where \( n_i = |U_i| \) and \( n = n_1 + n_2 \). Since $G$ is balanced, $L$ is non-singular, hence the kernel of $\mc L$ is two-dimensional and spanned by the lifts \( \mathbf{1}_{2n} \) and \( (\mathbf{u}, \mathbf{u})^\top \). The Fiedler vector is then \( (\mathbf{u}, \mathbf{u})^\top \), and its projection satisfies \( \f = \0 \) and
\(
\c = \sqrt{2} \sgn(\mathbf{u}).
\)
Since the entries of \( \mathbf{u} \) are constant over \( U_1 \) and \( U_2 \) with opposite signs, \( \c \) partitions \( V(G) \) according to the connected components, as claimed.
\end{proof}

\begin{proposition}
Let \( G  \) be a signed graph with \( k \) connected components, \( \ell \) of which are structurally balanced, and \( \mc L \) be the Laplacian of its Gremban expansion. Then the kernel of \( \mc L \) has dimension \( k + \ell \).
\end{proposition}

\begin{proof}
By Corollary \ref{cor:laplacian_spectrum_expansion}, the spectrum of \( \mc L \) is the union of the spectra of the \( L \) and \( \bar L \), hence $\dim \ker \mc L = \dim \ker L + \dim \ker \bar L$. Because $G$ has $k$ components, both $L$ and $\bar L$ can be block-diagonalized, with each block corresponding to the nodes of one component. For each block, $\bar L$ possesses an eigenvector with zero eigenvalue, whose entries are constant and equal to one on that component; hence $\dim \ker \bar L = k$. In addition, the corresponding block of $L$ is singular if and only if the component is balanced, giving $\dim \ker L = \ell$. Therefore, $\dim \ker \mc L = k+\ell$.
\end{proof}

\begin{algorithm}
\caption{Detecting $k$ clusters with the Gremban Laplacian}
\label{alg:community_faction_detection_multi}
\begin{algorithmic}[1]
\State \textbf{Input:} Signed graph \( G = (V,E,\sigma) \), number of clusters \( k \)
\State \textbf{Output:} Partition of \( V \) into communities and/or pairs of opposing factions
\State Compute the Laplacian matrix \( \mc L = \mc K - \mc A \) of $\mc G$ 
\State Compute the first \(  k-1 \) nontrivial eigenvectors of \( \mathcal{L} \) to form \( \mathcal{Y} \in \mathbb{R}^{2n \times (k-1)} \)
\State Cluster the rows of \( \mathcal{Y} \) using \( k \)-means to obtain \( \mathcal{U}_1, \dots, \mathcal{U}_k \)
\For{each cluster \( \mathcal{U}_i \)}
    \If{\( \eta(\mathcal{U}_i) = \mathcal{U}_i \)}
        \State Project \( \mathcal{U}_i \) to \( G \) using $\pi_+$ and interpret as a {community}
    \ElsIf{\( \eta(\mathcal{U}_i) = \mathcal{U}_j \) for some \( j \ne i \)}
         \State Project \( \mathcal{U}_i \) and \( \mathcal{U}_j \) to \( G \) using \( \pi_+ \)
        \State Interpret \( \pi_+(\mathcal{U}_i) \) and \( \pi_+(\mathcal{U}_j) \) as a pair of opposing {factions}
    \EndIf
\EndFor
\end{algorithmic}
\end{algorithm}
%
\section{Example of a Gremban expansion}  \label{app:triangle_example}
\begin{figure}[h!]
    \centering
    \includegraphics[width=0.6\linewidth]{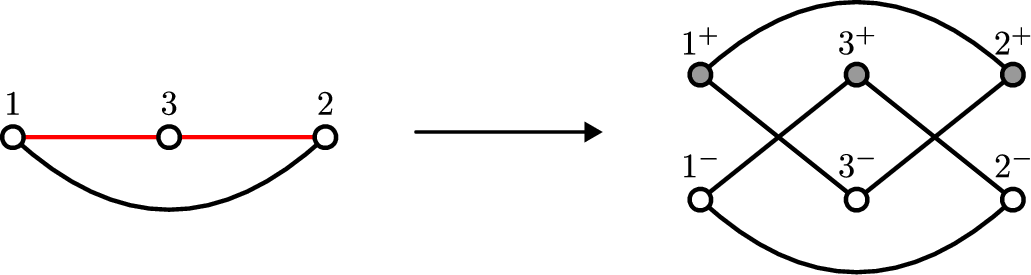}
    \caption{A triangle graph with one positive edge $(1,2)$ and two negative edges $(1,3)$ and $(2,3)$, and its Gremban expansion.}
    \label{fig:triangle_gremban}
\end{figure}

As a concrete example of the use of the Gremban expansion, consider the signed triangle graph shown in Figure \ref{fig:triangle_gremban} with nodes \(\{1,2,3\}\), edges $\{12,23,13\}$ and signs
\(
\sigma(12)=+1, \sigma(23)=-1, \sigma(13)=-1.
\)
Its signed adjacency matrix is
\[
A =
\begin{pmatrix}
0 & 1 & -1\\
1 & 0 & -1\\
-1 & -1 & 0
\end{pmatrix} = 
\begin{pmatrix}
0 & 1 & 0\\
1 & 0 & 0\\
0 & 0 & 0
\end{pmatrix} - 
\begin{pmatrix}
0 & 0 & 1\\
0 & 0 & 1\\
1 & 1 & 0
\end{pmatrix} ,
\]
and hence the Gremban expansion has adjacency matrix
\[
\mc A
= 
\begin{pmatrix}
0 & 1 & 0 & \vrule & 0 & 0 & 1\\
1 & 0 & 0 & \vrule & 0 & 0 & 1\\
0 & 0 & 0 & \vrule & 1 & 1 & 0\\\hline
0 & 0 & 1 & \vrule & 0 & 1 & 0\\
0 & 0 & 1 & \vrule & 1 & 0 & 0\\
1 & 1 & 0 & \vrule & 0 & 0 & 0
\end{pmatrix},
\]
where the vertical and horizontal rules separate the blocks \(A^+\) (upper-left and lower-right) and \(A^-\) (upper-right and lower-left).

Set
\[
\Pi_{s} =\frac1{\sqrt2}\ (\1_3\;\1_3),
\quad
\Pi_{a}   =\frac1{\sqrt2}\ (\1_3\;-\1_3),
\quad
\mc U   = \begin{pmatrix}\Pi_{s}\\[2pt]\Pi_{a}\end{pmatrix}
= \frac1{\sqrt2}
\begin{pmatrix}
\1_3 & \1_3\\
\1_3 & -\ \1_3
\end{pmatrix}.
\]
A direct computation gives:

\begin{align}
\mc U \mc A \mc U^\top = 
\begin{pmatrix}
0 & 1 & 1 & 0 & 0 & 0\\
1 & 0 & 1 & 0 & 0 & 0\\
1 & 1 & 0 & 0 & 0 & 0\\
0 & 0 & 0 & 0 & 1 & -1\\
0 & 0 & 0 & 1 & 0 & -1\\
0 & 0 & 0 & -1 & -1 & 0
\end{pmatrix},
\end{align}
where the last matrix is \(\begin{pmatrix}
\bar A & 0\\
0      & A
\end{pmatrix}\), just as claimed in Theorem \ref{thm: spectral properties of Gremban symmetric matrix}. 

The adjacency matrices have spectrum
\(
\spec(\bar A)=\{2,-1,-1\},
\
\spec(A)=\{2,-1,-1\},
\)
and
\(
\spec(\mc A)
=\{2,2,-1,-1,-1,-1\},
\)
where the latter is the (multiset) union of the former two, as expected from Theorem \ref{thm: spectral properties of Gremban symmetric matrix}. For the eigenvectors, we have:
\begin{itemize}
  \item For \(\bar A\): one eigenvector is \((1,1,1)^\top\) with eigenvalue \(2\).  It lifts to
  \(
    (1,1,1, 1,1,1)^\top\) in \(\mc A. \)
  The other two eigenvectors of \(\bar A\), with eigenvalue $-1$, are \((1,0,-1)^\top\) and \((0,1,-1)^\top\) (or any linear combination of them), which lift to
  \(
    (1,0,-1,\ 1,0,-1)^\top \) and \(
    (0,1,-1,\ 0,1,-1)^\top\) for $\mc A$.
  \item For \(A\): the eigenvector with eigenvalue \(2\) is \((1,1,-1)^\top\), which lifts antisymmetrically to
  \(
    (1,1,-1, -1, -1,\allowbreak 1)^\top.
  \)
  The eigenvectors with eigenvalue \(-1\) are \((1,0,1)^\top\) and \((0,1,1)^\top\), lifting to
  \(
    (1,0,1,\ -1,0,-1)^\top\) and \((0,1,1,\ 0,-1,-1)^\top.
  \) 
\end{itemize}
The Laplacian matrix of the Gremban expansion is
\[
\mc L = \mc K - \mc A =
\begin{pmatrix}
2 & -1 & 0 & \vrule & 0 & 0 & -1 \\
-1 & 2 & 0 & \vrule & 0 & 0 & -1 \\
0 & 0 & 2 & \vrule & -1 & -1 & 0 \\
\hline
0 & 0 & -1 & \vrule & 2 & -1 & 0 \\
0 & 0 & -1 & \vrule & -1 & 2 & 0 \\
-1 & -1 & 0 & \vrule & 0 & 0 & 2
\end{pmatrix}.
\]
If we apply the same change of basis $\mc U$ to $\mc L$, we get:
\begin{align}
\mc U \mc L \mc U^\top = 
\begin{pmatrix}
2 & -1 & -1 &  & 0 & 0 & 0 \\
-1 & 2 & -1 &  & 0 & 0 & 0 \\
-1 & -1 & 2 &  & 0 & 0 & 0 \\
0 & 0 & 0 &  & 2 & -1 & 1 \\
0 & 0 & 0 &  & -1 & 2 & 1 \\
0 & 0 & 0 &  & 1 & 1 & 2
\end{pmatrix}.
\end{align}
The last matrix is \( \left(\begin{smallmatrix} \bar L & 0\\
0      & L
\end{smallmatrix}\right)\), with $L$ the signed Laplacian of the network and $\bar L$ the unsigned Laplacian (i.e. the Laplacian of the underlying unsigned network). \\
With regard to the spectrum, $\mc L$ has a kernel of dimension two, spanned by $(1,1,1,1,1,1)^\top$ and $(1,1,-1,-1, \allowbreak -1,1)^\top$. Note that the second eigenvector is the antisymmetric lift of the vector indicating the balanced factions of the network. In addition, any vector simultaneously orthogonal to the previous two eigenvectors is an eigenvector of $\mc L$, with an associated eigenvalue of 3.
%

\section{Spectral clustering of unsigned graphs}  \label{app:spectral_clustering}

In this Appendix, we briefly review the theoretical basis for spectral clustering in unsigned graphs. The central idea is to detect communities by minimizing the number of edges between them. More formally, this involves finding a node partition that minimizes the size of the associated cut-set (Definition \ref{def:cut_set}), i.e., the total weight of edges connecting different parts of the graph. The minimum-cut problem admits a relaxation in terms of the eigenvectors of the graph Laplacian, particularly the one corresponding to the second smallest eigenvalue, called the Fiedler vector. 

Recall that, for an unsigned graph \(G=(V,E)\) with adjacency matrix \(A\) and degree matrix \(K\), the (combinatorial) Laplacian is defined as \(L=K-A\). The matrix \(L\) is symmetric and positive semidefinite, since \(\mathbf{x}^\top L\mathbf{x}=\tfrac12\sum_{vw}A_{vw}(x(v)-x(w))^2\ge0\) for all \(\mathbf{x}\in\R^n\). Moreover, \(L\mathbf{1}=\0\), so zero is always an eigenvalue; its multiplicity equals the number of connected components of \(G\), with eigenvectors given by the indicator vectors of the components.
%
\subsection{Two-communities case}
Now we prove that the Laplacian eigenvectors provide relevant information about the graph communities. We denote the complement of a set $S$ by $S^c$.

\begin{theorem}  \label{thm:min_cut_Laplacian}
Let $G=(V,E)$ be an unsigned graph with Laplacian $L$. Then the discrete minimum-cut problem  
\(
\min_{S\subset V}\;|\mathrm{C}(S,S^c)|
\)  
can be relaxed to the continuous problem  
\[
\min_{\mathbf{x}\perp\mathbf{1}}
\frac{\mathbf{x}^\top L\,\mathbf{x}}{\mathbf{x}^\top\mathbf{x}}
=
\min_{\mathbf{x}\perp\mathbf{1},\,\|\mathbf{x}\|=1}
\mathbf{x}^\top L\,\mathbf{x},
\]  
whose unique minimizer is the Fiedler eigenvector $\psi_2$ of $L$, i.e., the eigenvector corresponding to its second-smallest eigenvalue $\lambda_2$.
\end{theorem}

\begin{proof}
Let $(S,S^c)$ be an arbitrary bipartition of the node set $V$, and define the indicator vector \(\mathbf{s}\) by
\[
s(v) = 
\begin{cases}
+1,&v\in S,\\
-1,&v\notin S.
\end{cases}
\]
Then for any edge \((v,w)\in E(G)\),  
\[
(s(v) - s(w))^2
=\begin{cases}
0,&\text{if \(v,w\) lie on the same side of the cut},\\
4,&\text{if \(v,w\) lie on opposite sides}.
\end{cases}
\]
Hence
\(
\sum_{vw}A_{vw}\,(s(v) - s(w))^2
=4\,|C(S,S^c)|.
\) On the other hand, by the definition of the Laplacian:
\[
\mathbf{s}^\top L\,\mathbf{s} = \sum_{v}k_{v}s(v)^2 - \sum_{vw}A_{vw}s(v)s(w)
=\frac1{2}\sum_{vw}A_{vw}(s(v) - s(w))^2
=2\,|C(S,S^c)|.
\]
Therefore the discrete minimum-cut problem
\(
\min_{S\subset V}|C(S,S^c)|
\)  is equivalent to \(
\min_{\mathbf{s}\in\{\pm1\}^n}
\mathbf{s}^\top L\,\mathbf{s}.
\)

The discrete optimization domain \(\mathbf{s}\in\{\pm1\}^n\) can be relaxed to the convex optimization domain $\{\mathbf{x}\in\R^n:\Vert\mathbf{x}\Vert=1,\mathbf{x}\perp\mathbf{1}\}$. By the Courant-Fischer theorem, the minimum of \(\mathbf{x}L\mathbf{x}\) over this domain is exactly the second smallest eigenvalue of \(L\), and the minimizer is the corresponding eigenvector \(\psi_2\). 
\end{proof}

\begin{remark}  \label{rem:diffusion}
Theorem \ref{thm:min_cut_Laplacian} can be understood in terms of a diffusion process on the graph. The solution of the dynamical system
\(
\frac{d\mathbf{x}}{dt} = -L\,\mathbf{x},
\)
admits the spectral decomposition
\[
\mathbf{x}(t)
=\sum_{k=1}^n e^{-\lambda_k t}\,\langle \mathbf{x}(0),\,\psi_k\rangle\,\psi_k.
\]
The slowest-decaying modes are the ones associated with $\lambda_1=0$ and \(\lambda_2\), so for large \(t\) the concentration profile can be approximated as 
\(\mathbf{x}(t)\approx \langle \mathbf{x}(0),\mathbf{1}\rangle \mathbf{1} + \langle \mathbf{x}(0),\psi_2\rangle e^{-\lambda_2 t}\psi_2\).  In this ``metastable'' regime, nodes with similar entries in \(\psi_2\) share almost identical concentration values. Therefore the concentration of a diffusive process can reveal the community structure.
\end{remark}
%
\subsection{More than two communities}   
Spectral clustering can be easily extended to the case where there are $k$ communities. The only change is that, instead of only using the Fiedler vector, one uses the eigenvectors associated with the $k-1$ smallest nonzero eigenvalues of $L$, i.e., the set $\{\psi_2,...,\psi_k\}$. We first introduce an indicator matrix $H\in\{0,1\}^{n\times k}$, with entries
\[
H_{vr} = 
\begin{cases}
1,&\text{if node }v\text{ belongs to cluster }r,\\
0,&\text{otherwise},
\end{cases}
\]
with the constraint $H\,\mathbf{1}_k =\mathbf{1}_n$, so that each node lies in exactly one cluster. As in the proof of Theorem \ref{thm:min_cut_Laplacian}, with $H$ playing the role of $\mathbf{s}$, one can easily prove that $(H^\top L H)_{rr}=\frac12 \sum_{vw}A_{vw}(H_{vr}-H_{wr})^2$. The term $H_{vr}-H_{wr}$ is nonzero if and only if $v$ belongs to subset $r$ and $w$ does not, or vice versa. Hence, the matrix product $H^\top L H$ quantifies the total cut-set of a $k$-partition as
\(
\mathrm{tr}\bigl(H^\top L\,H\bigr) = \sum_{r=1}^k |C(S_r,S_r^c)|.
\)
The discrete \(k\)-way minimum-cut problem can be then rephrased as
\(
\min_{H}
\mathrm{tr}\bigl(H^\top L\,H\bigr)
\), with the constraints $H\in\{0,1\}^{n\times k},\;H\mathbf{1}_k=\mathbf{1}_n$. If we relax \(H\) to a real matrix \(X\in\R^{n\times k}\) with orthonormal columns, \(X^\top X = \1_k\), the resulting optimization problem becomes analytically solvable:

\begin{proposition}[\cite{horn2012matrix}]\label{prop:multi_k_spectal_clustering}
Let $G$ be an unsigned graph with Laplacian \(L\).  Then
\(
\min_{X\in\mathbb{R}^{n\times k},\;X^\top X = \1_k}
\;\tr\!\bigl(X^\top L X\bigr)
=\sum_{i=1}^k\lambda_i,
\)
and the minimum is attained when the column-space of \(X\)
is \(\mathrm{span}\{\psi_2,\dots,\psi_k\}\).
\end{proposition}

\begin{proof}
Write \(X\in\R^{n\times k}\) with \(X^\top X=\1_k\).  Define the symmetric matrix
\(
M = X^\top L X,
\)
with eigenvalues \(\mu_1\le\cdots\le\mu_k\).  Then
\(
\tr(X^\top L X)
=
\tr(M)
=
\sum_{j=1}^k \mu_j.
\)
By the Cauchy interlacing theorem,
\(\lambda_j \le \mu_j\) for \(j=1,\dots,k\), so
\(\sum_{j=1}^k \mu_j \ge \sum_{j=1}^k \lambda_j\).  Thus
\(
\tr(X^\top L X)
\ge
\sum_{j=1}^k \lambda_j
\).
On the other hand, choosing
\(X_* = (\psi_1,\psi_2,\dots,\psi_k)\) gives
\(X_*^\top X_* = \1_k\) and
\(
\tr(X_*^\top L X_*)
=\sum_{i=1}^k \psi_i^\top L \psi_i
=\sum_{i=1}^k \lambda_i.
\)
Therefore
\(\min_{X^\top X = \1}\tr(X^\top L X)=\sum_{i=1}^k\lambda_i\), attained if and only if
\(\textrm{Col}(X)=\textrm{span}\{\psi_1,\dots,\psi_k\}\).
\end{proof}
Proposition \ref{prop:multi_k_spectal_clustering} shows that the relaxed minimizer of the $k$-way cut objective is obtained by taking $X$ to be the matrix of the $k$ smallest eigenvectors of $L$.  In practice one discards the constant eigenvector \(\psi_1 = \frac{1}{\sqrt n}\mathbf1,\) and retains only the $k-1$ eigenvectors \(\psi_2,\dots,\psi_k\).

\subsection{Spectral embedding}

When the matrix $X$ in the clustering problem is allowed to take arbitrary real values, its columns no longer have a direct combinatorial interpretation as ``hard''cluster-assignments.  To address this, one typically embeds the graph into a Euclidean space where standard clustering tools become applicable. This process is commonly known as spectral embedding, and amounts to interpreting the components of the $k-1$ smallest eigenvectors of the Laplacian as coordinates in a $k-1$-dimensional space. More specifically, the embedding of node $v$ is defined as
\(
\y_v = \bigl(\psi_2(v), \psi_3(v), \dots, \psi_k(v)\bigr) \in \R^{k-1},
\)
where $\psi_j(v)$ denotes the $v$-th entry of the $j$-th eigenvector of the Laplacian. Equivalently, the embedding is the $v$-th row of the matrix
\(
Y = (\psi_2\, \psi_3\, \cdots\, \psi_k) \in \R^{n \times (k-1)}.
\)

This embedding ensures that nodes which are well-connected in the graph are mapped to nearby points in Euclidean space. This behavior follows directly from the fact that Laplacian eigenvectors define smooth functions on the graph:

\begin{proposition}[\cite{horn2012matrix}]
Let $G$ be a graph with Laplacian $L$. Then its eigenvectors admit the variational formulation
\[
\psi_j = \textrm{argmin}_{\substack{\u \in \R^n \\ \u \perp \psi_1,\dots,\psi_{j-1} \\ \|\u\|=1}} \u^\top L \u, \quad j = 1, \dots, n.
\]
In particular, the eigenvectors $\psi_2, \dots, \psi_k$ minimize the total Dirichlet energy
\(
E[\{\u_2,...,\u_k\}] = \sum_{i=2}^k \sum_{vw}  A_{vw} \times \allowbreak (u_i(v) - u_i(w))^2
\)
among all orthonormal sets \( \{\u_2, \dots, \u_k\} \subset \mathbb{R}^n \) orthogonal to \( \psi_1 \).
\end{proposition}

\begin{proof}
Since $L$ is real symmetric and positive semidefinite, the Courant--Fischer theorem gives the variational characterization of its eigenvalues:
\[
\lambda_j = \min_{\substack{U \subset \R^n \\ \dim(U) = j}} \max_{\substack{\u \in U \\ \|\u\| = 1}} \u^\top L \u 
\]
This minimum is attained when \( U = \operatorname{span}(\psi_1, \dots, \psi_j) \), and within this subspace, the maximum is achieved at \( \psi_j \), which is orthogonal to \( \psi_1, \dots, \psi_{j-1} \). Therefore, the characterization can be rewritten as
\[
\lambda_j = \min_{\substack{\u \perp \psi_1,\dots,\psi_{j-1} \\ \|\u\| = 1}} \u^\top L \u,
\]
with the minimum attained at $\psi_j$. This proves the first statement. To prove the second, we employ the quadratic form of the Laplacian
\(\u^\top L \u = \frac1{2} \sum_{vw} A_{vw} (u(v)-u(w))^2 
\)
to write the Dirichlet energy as
\(
E[\{\u_2, \dots, \u_k\}] = 2 \sum_{i=2}^k \u_i^\top L \u_i.
\)
The variational characterization then implies that $E$ is minimized under the constraint $u_i \perp \psi_1$ when \( \u_i = \psi_i \) for \( i = 2, \dots, k \).
\end{proof}
As a basic guarantee, we can show that spectral embedding exactly recovers the connected components when the graph is disconnected.
\begin{proposition}
Suppose $G$ consists of $k$ disconnected components $G_1,\dots,G_k$, and let $L$ be its Laplacian. Then the spectral embedding \( Y = [\psi_2 \ \dots \ \psi_k] \in \mathbb{R}^{n \times (k-1)} \) assigns all nodes in each component $G_r$ to a single point in $\R^{k}$.
\end{proposition}

\begin{proof}
When $G$ has $k$ connected components, $L$ has exactly $k$ zero eigenvalues. The eigenvectors $\psi_1,\dots,\psi_k$ can be chosen as orthonormal indicator vectors of the components, so that each $\psi_j$ is constant on $G_j$ and zero elsewhere. Under this choice, the spectral embedding then maps all nodes in $G_j$ to the same vector in $\R^{k}$.
Any other orthonormal basis of the null space of \( L \) is related to this one by a orthogonal transformation. Since these transformations preserve pairwise Euclidean distances, the geometry of the embedded points (and in particular, the fact that the image of each component collapses to a single point) remains invariant. Thus, any distance-based clustering algorithm, such as \( k \)-means, will correctly recover the connected components, regardless of the specific choice of orthonormal eigenvectors.
\end{proof}
%
\subsection{Algorithmic summary}
To cluster a graph into $k$ communities:
\begin{enumerate}
  \item Compute the first $k$ eigenpairs $(\lambda_j,\psi_j)$ of the Laplacian $L$.
  \item Form the embedding matrix 
  \(
    Y = \bigl(\psi_2,\psi_3,\dots,\psi_k\bigr)\in\R^{n\times (k-1)}.
  \)
  \item Treat each node $v$ as the point $\y_v\in\R^{k-1}$ given by the $v$-th row of $Y$.
  \item Run $k$-means on $\{\y_v\}_{v=1}^n$ to obtain clusters $U_1,\dots,U_k$.
\end{enumerate}

\begin{remark}
When the graph has highly heterogeneous degrees, it is common to normalize each row of the embedding matrix $Y$
to unit length before clustering.  Equivalently, one may compute the first $k$ eigenvectors of the \emph{symmetric normalized Laplacian}
\(
L^{\mathrm{norm}} = K^{-1/2}\,L\,K^{-1/2},
\)
and then embed and normalize as above.  This procedure corresponds to optimizing the \emph{normalized-cut} criterion and often improves cluster quality on graphs with very irregular degree distributions. In this work, we employ the normalized version of the Laplacians in the numerical simulations of Figures \ref{fig:interpolation_communities_factions}, \ref{fig:Spectrum_network_with_communities} and \ref{fig:Spectrum_network_with_factions}.
\end{remark}
%

\section{Numerical simulation details}  

\subsection{Signed Stochastic Block Model}
\label{app:ssbm}
We generated undirected, signed networks with \(n\) nodes using a degree-corrected stochastic block model with two groups. Instead of a single block matrix, in the signed setting we have two block matrices \(\rho^+\) and \(\rho^-\), controlling the probabilities of positive and negative edge creation respectively. For simplicity, we assume that the within-group edge probabilities are the same for both groups, which means that each matrix is parametrized by two control parameters:
\[
\rho^+
=\begin{pmatrix}
\rho^+_{\mathrm{in}} & \rho^+_{\mathrm{out}}\\[3pt]
\rho^+_{\mathrm{out}} & \rho^+_{\mathrm{in}}
\end{pmatrix},
\quad
\rho^-
=\begin{pmatrix}
\rho^-_{\mathrm{in}} & \rho^-_{\mathrm{out}}\\[3pt]
\rho^-_{\mathrm{out}} & \rho^-_{\mathrm{in}}
\end{pmatrix}.
\]
To generate the network, we follow the following procedure. First, we assign each node \(v\) uniformly at random to a group \(g_v\) and assign a positive activity \(\theta_v\) to it. For simplicity, we chose $\theta_v=1$ to generate homogeneous networks; other distributions of $\theta_v$ would result in different degree distributions. Second, for every unordered pair \(\{v,w\}\), we compute the combined propensity
\(
\lambda_{vw}
=\theta_v\,\theta_w\bigl(\rho^+_{g_v,g_w}+\rho^-_{g_v,g_w}\bigr),
\)
then place an edge with probability
\(
P(\text{edge }vw)
=1-\exp(-\lambda_{vw}),
\)
so that small \(\lambda_{vw}\) yields \(P\approx\lambda_{vw}\) and large \(\lambda_{vw}\) approaches one. For each edge $(u,v)$, we set its sign to \(+1\) with probability
\(
{\rho^+_{g_v,g_w}}/({\rho^+_{g_v,g_w}+\rho^-_{g_v,g_w}}),
\)
and to \(-1\) otherwise.
%

\subsection{Clustering evaluation metrics: ARI and NMI}
\label{app:clustering_metrics}

To evaluate the quality of inferred partitions, we use two standard metrics: the Adjusted Rand Index (ARI) and the Normalized Mutual Information (NMI). Both compare a predicted clustering to a ground-truth partition, are invariant under label permutations, and take values in \([0,1]\), with 1 indicating perfect agreement and values near 0 indicating independence or random overlap.

\paragraph{Adjusted Rand Index (ARI)}  
The Rand Index (RI) measures the proportion of node pairs that are consistently grouped together or apart in both partitions. Let \(U\) and \(V\) be two partitions of \(n\) elements. Define:
\begin{itemize}
    \item \(a\): number of pairs in the same group in both \(U\) and \(V\),
    \item \(b\): number of pairs in different groups in both \(U\) and \(V\),
    \item \(c\), \(d\): number of pairs assigned to the same group in one partition but different in the other.
\end{itemize}
Then the Rand Index is
\(
\mathrm{RI} = \frac{a + b}{a + b + c + d}.
\)
The ARI corrects this score for random chance:
\(
\mathrm{ARI} = \frac{\mathrm{RI} - \mathbb{E}[\mathrm{RI}]}{1 - \mathbb{E}[\mathrm{RI}]},
\)
where \(\mathbb{E}[\mathrm{RI}]\) is the expected value under a random model.

\paragraph{Normalized Mutual Information (NMI)}  
The mutual information (MI) between two partitions \(U\) and \(V\) quantifies their shared information:
\(
I(U, V) = \sum_{i=1}^{k} \sum_{j=1}^{l} P(i, j) \log \left( \frac{P(i, j)}{P(i)\,P(j)} \right),
\)
where \(P(i)\) is the probability that a randomly chosen element belongs to cluster \(i\) in \(U\), \(P(j)\) is the same for cluster \(j\) in \(V\), and \(P(i, j)\) is the joint probability that a node is in cluster \(i\) in \(U\) and in \(j\) in \(V\).
 The normalized mutual information is then given by
\(
\mathrm{NMI}(U, V) = \frac{2 I(U, V)}{H(U) + H(V)},
\)
where \(H = - \sum_{i=1}^{k} P(i) \log P(i)\) is the entropy of a partition.

\subsection{Eigenvectors for communities and factions}  \label{app:eigenvectors}
\begin{figure}
    \centering
    \includegraphics[width=0.8\linewidth]{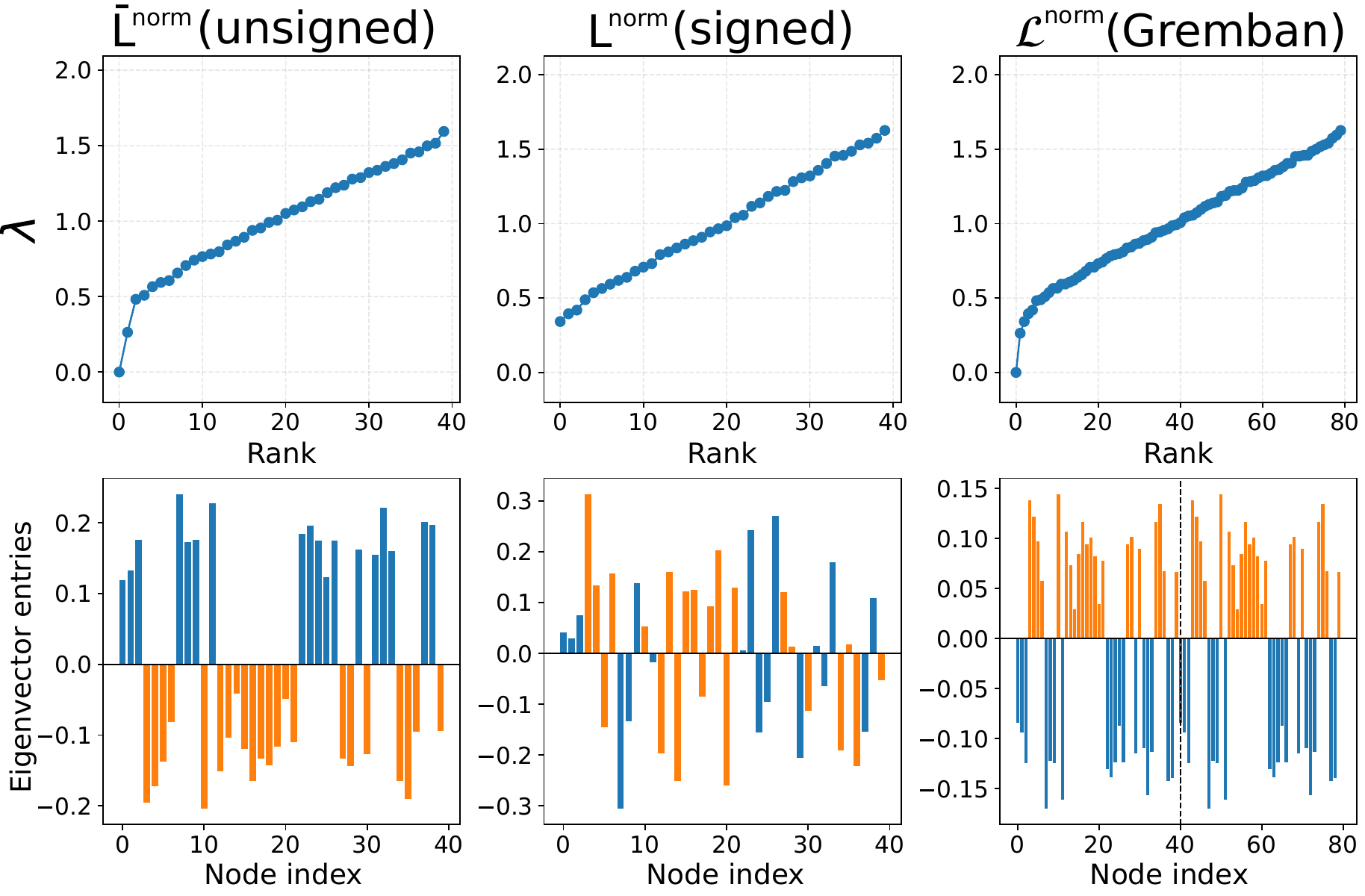}
    \caption{
    Spectral properties of three normalized Laplacian matrices for a network with community structure. This network is generated using the signed stochastic block model with \(n = 40\) nodes and two groups, and parameters: \(\rho^+_{\mathrm{in}} = 0.2\), \(\rho^+_{\mathrm{out}} = 0.02\), \(\rho^-_{\mathrm{in}} = 0.2\), \(\rho^-_{\mathrm{out}} = 0.02\). 
    Top row: eigenvalues of the unsigned Laplacian \(L^{\text{norm}}\), signed Laplacian \(\bar L^{\text{norm}}\), and Gremban Laplacian \(\mathcal{L}^{\text{norm}}\), ordered by increasing value. 
    Bottom row: components of the first non-constant eigenvector for each operator. Bar colors indicate the true group membership of each node. 
    }
    \label{fig:Spectrum_network_with_communities}
\end{figure}

\begin{figure}
    \centering
    \includegraphics[width=0.8\linewidth]{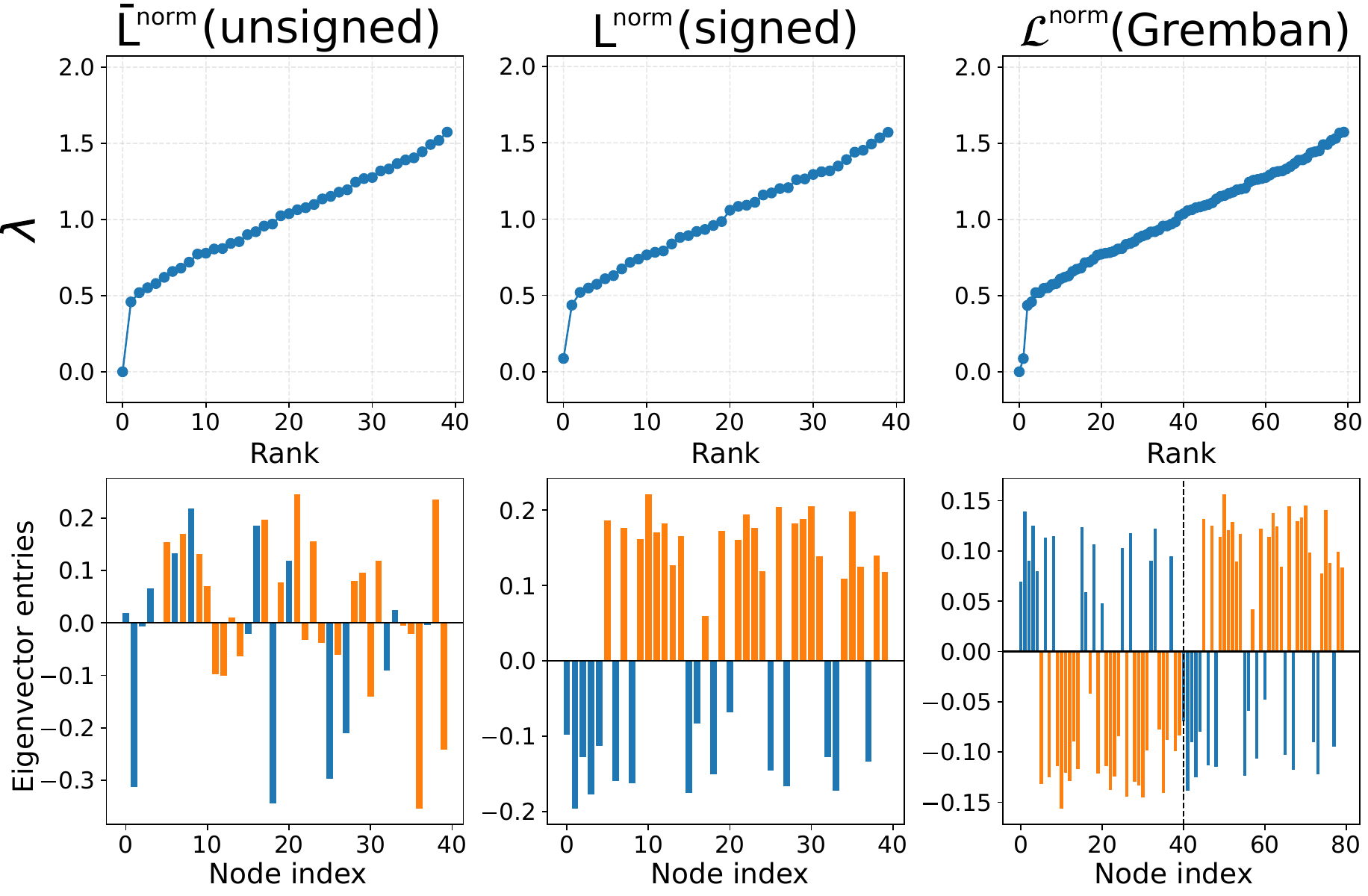}
     \caption{
    Spectral properties of the same three normalized Laplacian matrices for a network with faction structure, again generated with \(n = 40\) nodes and two groups. Parameters: \(\rho^+_{\mathrm{in}} = 0.2\), \(\rho^+_{\mathrm{out}} = 0.02\), \(\rho^-_{\mathrm{in}} = 0.02\), \(\rho^-_{\mathrm{out}} = 0.2\). 
    As in the previous figure, the top row shows the spectrum and the bottom row the first non-constant eigenvector. 
    }
    \label{fig:Spectrum_network_with_factions}
\end{figure}

Figures~\ref{fig:Spectrum_network_with_communities} and~\ref{fig:Spectrum_network_with_factions} provide a detailed spectral analysis that supports the use of the Gremban expansion as the right method to simultaneously handle communities and factions.  The figures show the spectra and leading eigenvectors of three normalized Laplacian matrices, 
\(L^{\text{norm}}\), \(\bar L^{\text{norm}}\) and \(\mathcal{L}^{\text{norm}}\), generated from a signed stochastic block model with \(n = 40\) nodes and two equally sized groups (Section~\ref{app:ssbm} for details). To control for degree heterogeneity, we use the normalized version of the Laplacians: $L^{\text{norm}}=  K^{-1/2}LK^{-1/2}$, $\bar L^{\text{norm}}=  K^{-1/2}\bar LK^{-1/2}$, and $\mc L^{\text{norm}}=  \mc K^{-1/2}\mc L \mc K^{-1/2}$.
Each figure shows the eigenvalues (top row) and the entries of the first non-constant eigenvector (bottom row), with bar colors indicating ground-truth group membership. 

In the community-dominant case, all eigenvalues of \(L^{\text{norm}}\) are larger than the first two eigenvalues of \(\bar L^{\text{norm}}\). The Fiedler vector of \(\bar L^{\text{norm}}\) perfectly separates the groups: the sign of each component matches the true label of the corresponding node. In contrast, the eigenvector of \(L\) performs no better than chance. For \(\mathcal{L^{\text{norm}}}\), the leading eigenvector is symmetric (the second half is identical to the first), indicating the presence of communities in the original graph. In addition, once projected back to the original node set, the sign pattern correctly classifies all nodes.
In the faction-dominant case, the smallest eigenvalue of \(L^{\text{norm}}\) lies between the first and second eigenvalues of \(\bar L^{\text{norm}}\). The associated eigenvector of $L^{\text{norm}}$ recovers the faction structure, while the Fiedler vector of \(\bar L^{\text{norm}}\) again yields classifications not better than chance. The Fiedler eigenvector of \(\mathcal{L}^{\text{norm}}\) is antisymmetric, consistent with the presence of factions in the original graph. Projecting this vector onto the original node set yields a perfect recovery of the underlying faction structure. These results confirm that \(\mathcal{L}^{\text{norm}}\) robustly adapts to the underlying mesoscale organization, whether community-like or faction-like.

\section{Walk enumeration via Gremban expansions}   \label{app:walk_enumeration}

One further use of the Gremban expansion is that it separates the enumeration of positive and negative walks through standard matrix operations. Recall that a walk is a sequence of (not necessarily distinct) edges \( W = (e_{12}, e_{23}, \ldots, e_{k-1,k}) \) where consecutive edges are incident to the same node. The sign of a walk is the product of the sign of its edges: $\sigma(W)=\prod_{e\in W} \sigma(e)$. Consequently, a walk is positive when it contains an even number of negative edges, and negative when it contains an odd number.
We begin by recalling how to count walks in the original signed graph:

\begin{proposition}  \label{prop:walk_adjacency_lemma}
Let \(G=(V,E,\sigma)\) be a signed graph with signed adjacency matrix \(A\) and unsigned adjacency matrix \(\bar A=|A|\). Let $[P^{(k)}]^+_{vw}$ denote the number of positive walks 
of length $k$ from $v$ to $w$, and $[P^{(k)}]^-_{vw}$ denote the number of negative walks of length $k$ from $v$ to $w$. For any integer \(k\ge1\) and any pair of nodes \(v,w\in V\):
\[
(A^k)_{vw}
=[P^{(k)}]^+_{vw}  -  [P^{(k)}]^-_{vw},
\qquad 
(\bar A^k)_{vw}
=[P^{(k)}]^+_{vw} +  [P^{(k)}]^-_{vw}
\]
\end{proposition}

\begin{proof}
We proceed by induction on the walk length $k$. \\ \noindent\emph{Base case:}  By definition, \(A_{vw}=1\) for a positive edge, \(-1\) for a negative edge, and \(0\) otherwise. Thus \(A_{vw}\) encodes whether the unique walk of length~1 is positive, negative, or absent, while \(\bar A_{vw}=|A_{vw}|\) counts it regardless of sign.

\medskip

\noindent\emph{Induction step:}  Assume the statement holds for walks of length \(k-1\).  Then
\(
(A^k)_{vw}
=\sum_{m}(A^{k-1})_{vm}\;A_{m j}.
\)
By the induction hypothesis, \((A^{k-1})_{vm}\) equals
\(
[P^{(k-1)}]^+_{vw}
-
[P^{(k-1)}]^-_{vw}\)
and \(A_{mw}=\pm1\) or \(0\) encodes the sign of the final step \(m\to j\).  Thus each product
\((A^{k-1})_{vm}\,A_{mw}\)
equals \(+1\) for each positive length-\(k\) walk through \(m\), \(-1\) for each negative one, and \(0\) if no edge.  Summing over all \(m\) therefore exactly computes
\([P^{(k)}]^+_{vw}
-
[P^{(k)}]^-_{vw}\), as claimed.  
The same argument applies to \(\bar A\), since
\(
(\bar A^k)_{vw}
=\sum_{m}(\bar A^{k-1})_{vm}\;\bar A_{mw},
\)
and now \(\bar A_{mw}=|A_{mw}|\in\{0,1\}\) counts each walk regardless of sign, so the sum over all intermediate \(m\) gives the total number of walks of length \(k\) from \(v\) to \(w\).
\end{proof}
The previous proposition shows how to count sums or differences of walks using powers of the signed adjacency matrix. It is often more convenient, though, to work with an operator whose powers directly enumerate positive and negative walks separately. This operator is precisely the Gremban expansion of the signed adjacency matrix. To show this, we start with the following lemma: 

\begin{lemma}
    \label{lem:paths_expanded_graph}
Let $G$ be a signed graph with Gremban expansion $\mc G$. Fix any $v,w\in V(G)$ and any polarity $\chi\in\{+,-\}$.  Then for every $v$-$w$ walk $W$ in $G$ there is a corresponding walk $\mc W$ in $\mc G$ from $v^{\chi}$ to $w^{\chi \cdot \sigma(W)}$.
\end{lemma}

\begin{proof}
Let 
\(
W  =  v_0\ v_1 \cdots v_\ell
\)
be a walk in $G$ with $v_0=v$ and $v_\ell=w$.
Define signs \(\chi_0,\chi_1,\dots,\chi_\ell\in\{+,-\}\) by 
\(
\chi_0=\chi,
\
\chi_i = \chi_{i-1} \sigma(v_{i-1}v_i)
\), so that $\chi_\ell = \chi_0 \cdot \sigma(P)$.
Then, for each \(i\) the lifted edge
\(
e_i = (v_{i-1}^{\chi_{i-1}},\ v_i^{\chi_i})
\)
is present in the Gremban expansion. Therefore $v_0^{\chi_0}\ v_1^{\chi_1}\hdots v_\ell^{\chi_\ell} $ form a walk in $\mc G$.
\end{proof}
Lemma~\ref{lem:paths_expanded_graph} shows that every walk \(W\) in \(G\) gives rise to exactly two walks in \(\mc G\). Concretely, a \textit{positive} \(v\)-\(w\) walk in \(G\) corresponds to one walk in \(\mc G\) from \(v^+\) to \(w^+\) and another from \(v^-\) to \(w^-\). A \textit{negative} \(v\)-\(w\) walk in \(G\) corresponds to one walk from \(v^+\) to \(w^-\) and another from \(v^-\) to \(w^+\). This observation leads directly to the following\footnote{This result was also proved in \cite[Theorem~3.4]{fox2017numerical}; however, we believe our derivation is simpler and more transparent.}:

\begin{corollary}  \label{cor:walk_adjacency_expanded}
Let $G$ be a signed graph with adjacency matrix $A$ and $\mc G$ be its Gremban expansion, with adjacency matrix $\mc A$. For any \(v,w\in V(G)\), the number of positive walks of length \(k\) from \(v\) to \(w\) equals the number of walks of length \(k\) from \(v^+\) to \(w^+\), namely the entry \([\mc A^k]^{++}_{vw}\). Likewise, the number of negative walks of length \(k\) from \(v\) to \(w\) equals the number of walks from \(v^+\) to \(w^-\), given by \([\mc A^k]^{+-}_{vw}\).
\end{corollary}
We stress that neither $A^k$ nor $\bar A^k$ alone contains complete information about the signed walk structure, while the expanded adjacency matrix $\mc A^k$ does, as its block structure explicitly separates the contributions from positive and negative walks. Note that, using Corollary \ref{cor:walk_adjacency_expanded}, Proposition \ref{prop:walk_adjacency_lemma} can be rewritten as:
\begin{align} \label{eq:Ak_walks}
A^k = [\mc A^k]^{++} - [\mc A^k]^{+-}, \quad \bar A^k = [\mc A^k]^{++} + [\mc A^k]^{+-},
\end{align}
which in turn lets us express the expanded adjacency matrix $\mc A$ as: 
\begin{align} \label{eq:expanded_Ak}
    \mc A^k = \frac{1}{2}
    \begin{pmatrix}
        \bar A^k + A^k & \bar A^k - A^k \\
        \bar A^k - A^k & \bar A^k + A^k
    \end{pmatrix}.
\end{align}
Interestingly, this identity could also be derived via spectral arguments. Since \( \mc A = \mc U(\bar A \oplus A) \mc U^\top \) (Theorem \ref{thm: spectral properties of Gremban symmetric matrix}), we immediately get $\mc A^k = \mc U(\bar A^k \oplus A^k) \mc U^\top,
$ and expanding this product recovers Equation \eqref{eq:expanded_Ak}. Although this spectral derivation is more compact, the combinatorial perspective provides deeper insight into how the expansion encodes walk polarity through its structure.
%
\subsection{Walk-Generating Functions and the Gremban Expansion}  \label{app:walk_generating_functions}

Generating functions allow us to package the counts of walks of all lengths into a single analytic object, making spectral and asymptotic properties of the graph immediately accessible.  In the context of the Gremban expansion, they also neatly separate positive- and negative-signed walks.

\begin{definition}[Generating functions] \label{def:generating_function}
    Let $G$ be a signed graph with adjacency matrix \(A\) and $\mc G$ its Gremban expansion, with adjacency matrix $\mc A$.  For a formal parameter \(t\in \R\), the \emph{signed, unsigned, and expanded generating functions} of the graph are
\[
W_A(t)=\sum_{k=0}^\infty A^k\,t^k,\qquad
W_{\bar A}(t)=\sum_{k=0}^\infty  \bar A^k\,t^k,\qquad
W_{\mc A}(t)=\sum_{k=0}^\infty \mc A^k\,t^k.
\]
\end{definition}
Note that these series converge on the disks $|t|<\rho(A)^{-1}$, $|t|<\rho(\bar A)^{-1}$, and $|t|<\rho(\mc A)= \min\{\rho(A)^{-1},\,\rho(\bar A)^{-1}\}$, respectively. In each case, the function to which they converge is the resolvent $(\1-tM)^{-1}$ (where $M \in\{A, \bar A, \mc A\}$).
Now we connect these three generating functions:

\begin{proposition} \label{prop:generating_functions_connection}
Let $W_A, W_{\bar A}$ and $W_\mc A$ defined as in Definition \ref{def:generating_function}. Then, the following identity holds true:
\begin{align} \label{eq:generating_functions_connection}
    W_\mc A(t) = \frac1{2} \begin{pmatrix}
        W_A(t)+W_{\bar A}(t) &  W_A(t)-W_{\bar A}(t) \\
       W_A(t)-W_{\bar A}(t) &  W_A(t)+W_{\bar A}(t)
    \end{pmatrix}
\end{align} 
\end{proposition}

\begin{proof}
We will start from the identity 
\(\mc A^k=\frac12 \mc U(\bar A^k \oplus A^k) \mc U^\top\). Multiplying by $t^k$ and summing over $k$, one gets
\[
\sum_{k=0}^\infty \mc A^k\,t^k = \mc U \;\left(\sum_{k=0}^\infty  \bar A^k\,t^k \oplus
\sum_{k=0}^\infty A^k\,t^k\right) \; \mc U^\top.
\]
which immediately yields \(W_\mc A(t)= \mathcal{U} (W_{\bar A}(t) \oplus W_A(t))\mathcal{U}^\top\). Expanding the matrix multiplication yields Equation \eqref{eq:generating_functions_connection}.
\end{proof}
Let us now turn our attention to another  generating function widely used in the context of complex networks: the communicability matrix \cite{estrada2008communicability,diaz2025signed}.

\begin{definition}[Communicability]  \label{def:communicability}
    Let $G$ be a signed graph with adjacency matrix \(A\) and $\mc G$ its Gremban expansion, with adjacency matrix $\mc A$. For a formal parameter \(t\in \R\), we define the \emph{communicability matrix} as the following power series:
\[
\Gamma_A(t)=\sum_{k=0}^\infty A^k\,\frac{t^k}{k!}=\exp(A),\]
with equivalent definitions for $\bar A$ and $\mc A$.
\end{definition}
As opposed to the generating functions of Definition \ref{def:generating_function}, the communicability matrix converges regardless of the spectral radius of the original matrix. The resulting matrix function is the matrix exponential. This makes the communicability a common operator in the study of dynamical processes on networks. An analog of Proposition \ref{prop:generating_functions_connection} holds for the communicability:

\begin{proposition}
Let $\Gamma_A, \Gamma_{\bar A}$ and $\Gamma_\mc A$ 
be as in Definition \ref{def:communicability}. Then, the following identity holds true:
\begin{align*}
    \Gamma_\mc A(t) = \frac1{2} \begin{pmatrix}
        \Gamma_A(t)+\Gamma_{\bar A}(t) &  \Gamma_A(t)-\Gamma_{\bar A}(t) \\
       \Gamma_A(t)-\Gamma_{\bar A}(t) &  \Gamma_A(t)+\Gamma_{\bar A}(t).
    \end{pmatrix}
\end{align*}
\end{proposition}

\begin{proof}
    Identical to that of Proposition \ref{prop:generating_functions_connection}.
\end{proof}
\end{document}